\documentclass[compsoc,journal]{IEEEtran}
\usepackage{amsmath,amsfonts,amssymb}
\usepackage{algorithmic}
\usepackage{array}
\usepackage[caption=false,font=normalsize,labelfont=sf,textfont=sf]{subfig}
\usepackage{textcomp}
\usepackage{stfloats}
\usepackage{url}
\usepackage{verbatim}
\usepackage{graphicx}
\def\BibTeX{{\rm B\kern-.05em{\sc i\kern-.025em b}\kern-.08em
    T\kern-.1667em\lower.7ex\hbox{E}\kern-.125emX}}
\usepackage{balance}
\usepackage{cite}
\usepackage{xcolor}
\usepackage{tabularx}
\usepackage{algorithm}
\usepackage{tikz}
\usepackage{bm}
\usepackage[normalem]{ulem}
\newcolumntype{Y}{>{\small\raggedright\arraybackslash}X}

\DeclareMathOperator*{\argmax}{arg\,max}

\usepackage{amsthm}
\usepackage{multirow}
\usepackage{float}
\usepackage{subfig}
\usepackage{svg}
\usepackage{caption}
\usepackage{subcaption}
\usepackage{pifont}
\usepackage{setspace}
\usepackage{mathtools} % erweiterte Fassung von amsmath
\usepackage[nopostdot,acronym,shortcuts,nonumberlist]{glossaries}
\newacronym{sus}{SUS}{status update system}
\newacronym{mcs}{MCS}{mobile crowdsensing}
\newacronym{mu}{MU}{mobile unit}
\newacronym{mcsp}{MCSP}{mobile crowdsensing platform}
\newacronym{dr}{DR}{data requester}
\newacronym{cpu}{CPU}{central processing unit}
\newacronym{psnr}{PSNR}{peak signal-to-noise ratio}
\newacronym{aoi}{AoI}{age Of information}
\newacronym{iot}{IoT}{Internet of Things}
\newacronym{dt}{DT}{digital twin}
\newacronym{mwc}{MWC}{matching with contracts}
\usepackage{inputenc}
\usepackage{comment}
\usepackage{booktabs}
\usepackage[binary-units]{siunitx}

\newtheorem{theorem}{Theorem}
\newtheorem{assumption}{Assumption}
\newtheorem{definition}{Definition}

\begin{document}
\newcommand{\ao}[1]{\textcolor{magenta}{#1}}
\newcommand{\rc}[1]{\textcolor{red}{#1}}
\renewcommand{\algorithmiccomment}[1]{\quad{$\triangleright$\ #1}}

\title{Dynamic Hypergame for Task Assignment in Multi-platform Mobile Crowdsensing Under Incomplete Information}

\author{
\IEEEauthorblockN{Sumedh J. Dongare\IEEEauthorrefmark{1}, Christo Kurisummoottil Thomas\IEEEauthorrefmark{5}, Andrea Ortiz\IEEEauthorrefmark{3}, Walid Saad\IEEEauthorrefmark{4}, Anja Klein\IEEEauthorrefmark{1}}
\\
\IEEEauthorblockA{
\small
\IEEEauthorrefmark{1}Communication Engineering Lab, Technical University of Darmstadt, Germany.\\
\IEEEauthorrefmark{5}Dept. of Electrical and Computer Engineering, Worcester Polytechnic Institute, Worcester, MA, USA.\\
\IEEEauthorrefmark{3}Institute of Telecommunications, Vienna University of Technology, Austria.\\
\IEEEauthorrefmark{4}Wireless@VT, Bradley Department of Electrical and Computer Engineering, Virginia Tech, Alexandria, USA.\\
Emails: \{s.dongare,  a.klein\}@nt.tu-darmstadt.de, cthomas2@wpi.edu, andrea.ortiz@tuwien.ac.at, walids@vt.edu}
\thanks{This work was funded by the BMFTR project Open6GHub+ un-
der grant 16KIS2407, by DAAD with funds from the German Fed-
eral Ministry of Research, Technology and Space BMFTR under grant
57817830, and by the LOEWE Center emergenCITY under grant
LOEWE/1/12/519/03/05.001(0016)/72.
The work of Andrea Ortiz was funded by the Vienna Science and Technology Fund WWTF under grant 10.47379/VRG23002.
The work of Walid Saad was supported by the U.S. National Science Foundation under grant 2201641.
}
}
\maketitle
\newcommand{\MUindex}{k}
\newcommand{\setOfMUs}{\mathcal{K}}
\newcommand{\setOfTaskTypes}{\mathcal{Z}}
\newcommand{\numberOfTaskTypes}{Z}
\newcommand{\taskTypeIndex}{z}
\newcommand{\taskTypeZ}{\taskTypeIndex}
\newcommand{\taskIndex}{n}
\newcommand{\timeindex}{t}
\newcommand{\MCSPindex}{i}

\newcommand{\MUTaskAndTimeIndex}{_{\MUindex,\taskIndex,\timeindex}}
\newcommand{\MUTaskIndex}{_{\MUindex,\taskIndex}}

\newcommand{\sensingTime}{\tau^{\mathrm{sense}}\MUTaskAndTimeIndex}
\newcommand{\expectedSensingTime}{\bar{\tau}^{\mathrm{sense}}_{\MUindex,\taskTypeIndex}}
\newcommand{\communicationTime}{\tau^{\mathrm{comm},\MCSPindex}\MUTaskAndTimeIndex}
\newcommand{\computingTime}{\tau^{\mathrm{comp}}\MUTaskAndTimeIndex}
\newcommand{\expectedCommunicationTime}{\bar{\tau}^{\mathrm{comm},\MCSPindex}_{\MUindex,\taskTypeIndex}}
\newcommand{\totalEnergy}{E\MUTaskAndTimeIndex}
\newcommand{\sensingEnergy}{E^{\mathrm{sense}}\MUTaskAndTimeIndex}
\newcommand{\communicationEnergy}{E^{\mathrm{comm}}\MUTaskAndTimeIndex}
\newcommand{\computingEnergy}{E^{\mathrm{comp}}\MUTaskAndTimeIndex}

\newcommand{\totalTime}{\tau\MUTaskAndTimeIndex}
\newcommand{\expectedUtilityTotalTimeWithIndex}[1]{\bar{\tau}^{\mathrm{MU}}_{#1}}
\newcommand{\expectedTotalTime}{\expectedUtilityMUWithIndex{\MUindex,\taskIndex}}

\newcommand{\txPower}{p_k^\mathrm{comm}}
\newcommand{\sensePower}{p_{k,n}^\mathrm{sense}}
\newcommand{\compPower}{p_k^\mathrm{comp}}

\newcommand{\costFunction}{C^{\mathrm{effort}}_{k,n,t}}

\newcommand{\rewardTaskCompletion}{w_{\taskTypeIndex,\timeindex}}
\newcommand{\rewardPerTaskType}{w_{\taskTypeIndex}}
\newcommand{\MUpreferenceWithIndex}[1]{\succeq^{\mathrm{MU}}_{#1}}
\newcommand{\MUpreference}{\MUpreferenceWithIndex{\MUindex}}
\newcommand{\MUpreferenceTimeDependent}{\MUpreferenceWithIndex{\MUindex, \timeindex}}
\newcommand{\MUstrictPreferenceWithIndex}[1]{\succ^{\mathrm{MU}}_{#1}}
\newcommand{\MCSPTaskpreferenceWithIndex}[1]{\succeq^{\mathrm{MCSP}}_{#1}}
\newcommand{\MCSPTaskpreference}{\MCSPTaskpreferenceWithIndex{\MCSPindex}}

\begin{abstract}
Mobile crowdsensing (MCS) is a promising distributed sensing paradigm for future wireless networks, where MCS platforms (MCSPs) recruit mobile units (MUs) through monetary incentives for sensing data collection. While most existing studies assume a single MCSP, practical deployments involve multiple competing MCSPs that simultaneously propose task offers to MUs, and MUs accept offers that maximize their revenue. This interaction gives rise to a two-sided matching game with contracts (MWC), decomposed into two components: (i) task proposal problem of the MCSPs and (ii) task acceptance problem of the MUs. To optimally solve (i), every MCSP requires information about other platforms' preferences and the qualities of the MUs in advance. Similarly, to solve (ii) optimally, the MUs require information about the task execution efforts of all tasks in advance. Such information is unavailable at the MCSPs and at the MUs. To address the challenge of unknown preferences of the other MCSPs, the MWC is posed as a dynamic hypergame, where every MCSP models the unknown preferences through perceptions and refines them over repeated interactions. To solve the dynamic hypergame under incomplete information, we propose PACMAB, a fully decentralized perception-aware two-sided learning framework where, (i) each MCSP learns an adaptive task proposal strategy under competition, and (ii) each MU learns task acceptance policy by estimating task execution efforts. Computational complexity of PACMAB shows that it scales favorably for the MCSPs as well as the MUs. Extensive simulations show that PACMAB consistently outperforms the benchmarks by completing at least $41\%$ more tasks without assuming complete information.
\end{abstract}

\section{Introduction}
\subsection{Overview}
\IEEEPARstart{R}{}ecently, \gls{mcs} emerged as a promising distributed sensing alternative to traditional wireless sensor networks (WSNs) \cite{Ganti2011MCS_intro}.
Compared to WSNs, \glspl{mcs} provides lower infrastructure costs, higher mobility, better coverage, and a wide range of applications due to availability of various sensors on mobile units to perform sensing tasks \cite{Gong_2018_MCS_advantages, Dongare_EHMCS_2022, Dongare_Globecom_2023}.
With the advancements in the \gls{iot} and a rapidly growing number of smart devices, \glspl{mcs} has become a topic of interest in academia and in industrial applications \cite{MCS_intro_Dai_21, iot_mcs_Jian_2015} such as traffic \cite{Jiang_trafficMCS_2023} and environmental monitoring \cite{Dinh_environmental_monitoring_2022}, spectrum sensing \cite{MCS_for_spectrumsensing_Li2018,Spectrum_sensing_MCS_2021}, and mobile health \cite{Pryss2018mHealth}.

An \gls{mcs} architecture consists of \glspl{dr}, \glspl{mcsp}, and  \glspl{mu}.
When a \gls{dr} requires some sensing data from the target region, it creates a sensing request and sends it to an \gls{mcsp}.
The \gls{dr} offers a payment as an incentive to the selected \gls{mcsp}.
The \gls{mcsp} uses then part of this payment to recruit \glspl{mu} for the sensing by sending them task offers.
If an \gls{mu} accepts the offer from an MCSP, it performs the task, collects the sensing result, and transmits it back to the respective \gls{mcsp}.
The \gls{mu} receives the agreed payment according to the original task offer.
Every \gls{mcsp} decides its own task assignment strategy which maximizes its own net revenue.
The \glspl{dr} incentivize the \glspl{mcsp} to obtain sensing data with better quality by offering the \glspl{mcsp} payments proportional to the quality of sensing result.
For each sensing task, different \glspl{mu} may produce sensing results of different quality.
Equivalently, the quality of the sensing result varies for every \gls{mu}-task pair.
Every \gls{mu} decides on  which task offer to accept from which MCSP depending on the expected net revenue from the offers, i.e., the payment offered minus the task execution efforts.

\vspace{-0.5em}
\subsection{Research challenges}
\label{subsec:research_challenges}
The task assignment between the \glspl{mcsp} and the \glspl{mu} determines the success of the \gls{mcs} system.
Ideally, the assignment must maximize the revenues of both the \glspl{mcsp} and the \glspl{mu}, such that neither of the them have any incentive to deviate from the assignment.
To achieve such task assignment, the \gls{mcs} system has to overcome several challenges:\\
1) \textit{Conflicting interests of the \glspl{mcsp} and the \glspl{mu}}: 
The task proposal strategies of the \glspl{mcsp} and the acceptance strategies of the \glspl{mu}' are independent from each other and based only on their own net revenue.
For the \glspl{mcsp}, this means balancing between the revenue they get from the \glspl{dr} and the payment offered to the \glspl{mu}.
Similarly, the \glspl{mu} must balance between payments they receive and the efforts required to perform the tasks.
%Why this is a challenge?
Since the interests of the MCSPs and the MUs are not aligned, finding the optimal proposal and acceptance strategies is challenging.
\\
2) \textit{Competing \glspl{mcsp}}:
Since the number of available \glspl{mu} is finite, the \glspl{mcsp} compete with one another to have their proposals accepted.
Thus, every \glspl{mcsp} aims to make its task offers more attractive by selecting a suitable payment. % from a range of possible payments.
A higher payment increases the likelihood of \gls{mu}s acceptance, however, it reduces the net revenue of the \gls{mcsp}.
\\
3) \textit{Incomplete information}:
To make optimal decisions,  \glspl{mcsp} and \glspl{mu} require complete information about the \gls{mcs} system.
For the \glspl{mcsp}, this means information about (i) revenues earned from all \gls{mu}-task combinations, (ii) preferences of the other \glspl{mcsp}, and (iii) \glspl{mu}' preferences over the offered tasks.
Similarly, the \glspl{mu} must know the efforts required to perform the offered tasks.
However, in realistic scenarios, such information is unavailable at the \glspl{mcsp} and \glspl{mu}.
Thus, obtaining a task proposal strategy at the \glspl{mcsp} and a task acceptance strategy at the \glspl{mu} under incomplete information is crucial for the success of \gls{mcs}.

\vspace{-0.5em}
\subsection{Related works}
\label{subsec:related_works}
Depending on the information available to the decision-making entities, existing task assignment approaches in \gls{mcs} can be broadly classified into optimization-\cite{OPAT_Huang_2022, OEBS_2023_chang_LP_GA_optimization, RATE_optimization_2024_zhao, Xie_2024_optimization_comput_crowdsensing}, game-theory-\cite{Wang_2019_auction_game, Yucel_2022_game_STA, Bernd_ICC_2022}, and learning-based solutions\cite{Bernd_WSaad_OSL_2024, Dongare_EHMCS_2022, Dongare_Globecom_2023, C1:own:Dongare2024b}.
These approaches differ significantly in scalability, realism, and robustness.

Optimization-based solutions formulate task assignment as centralized profit, energy, or social-welfare maximization problems under spatial, temporal, and budget constraints \cite{OPAT_Huang_2022, OEBS_2023_chang_LP_GA_optimization, RATE_optimization_2024_zhao, Xie_2024_optimization_comput_crowdsensing}.
Such methods achieve near-optimal performance, but rely on complete non-causal information about the task characteristics and user capabilities.
Assuming availability of such information at the \gls{mu} or at the \gls{mcsp} is unrealistic.
Moreover, these works suffer from high computational complexity, which limits their applicability in large-scale and dynamic \gls{mcs} systems.

To improve scalability, game-theoretic approaches exploit decentralized decision-making by explicitly modeling strategic interactions between \glspl{mcsp} and \glspl{mu}, e.g., by modeling the problem as \gls{mwc} \cite{Wang_2019_auction_game, Yucel_2022_game_STA, Bernd_ICC_2022}.
However, these works assume availability of complete information regarding preferences of the \glspl{mcsp} and the \glspl{mu} which usually depends on task execution efforts, and data qualities of the \glspl{mu}.
Moreover, task characteristics and agent preferences are usually assumed to be static, which does not reflect the stochastic and time-varying nature of real-world \gls{mcs} systems.

Learning-based approaches relax these assumptions by allowing \glspl{mcsp} and \glspl{mu} to learn task utilities and preferences over time \cite{Bernd_WSaad_OSL_2024, Dongare_EHMCS_2022, Dongare_Globecom_2023, C1:own:Dongare2024b}.
Despite the improved realism, these works consider a single \gls{mcsp} handling multiple \glspl{mu}, thereby neglecting inter-platform competition and strategic interactions among multiple \glspl{mcsp}.
A more realistic multi-platform \glspl{mcs} has been recently investigated in \cite{Li_multiplatform_coop_comp_2021, Multiplatform_coop_2025_Peng, Multiplatform_coop_2025_Liu, Multiplatform_Yang_price_match_hybrid_2025, DNiyato_2025_MultiMCSP_stackelberg_MADDPG}.
However, these studies assume either an explicit cooperation among \glspl{mcsp} or the existence of a trusted and impartial cross-platform task management entity.
Such assumptions are difficult to fulfill in competitive real-world deployments, where \glspl{mcsp} act selfishly.

In an MCS system with multiple independent decision making agents, finding an optimal strategy for any agent would require information about the preferences of the other agents.
However, the availability of such information is unrealistic to assume in practical applications.
In such cases, hypergame theory \cite{BENNETT1977749, BENNETT1980489, BENNETT1980293Bidders_and_dispensers} provides a framework to model strategic interactions between decision-making agents under incomplete information.
Using this framework, agents maintain and update perceptions about the strategies or preferences of the other agents and use them to device a strategy under incomplete information.
A hypergame is dynamic when the perceptions change over time~\cite{dynamic_HG_2023_cybersecurity_deceptions}.
Hypergames have been applied in cyber-security and semantic communications \cite{dynamic_HG_2023_cybersecurity_deceptions, hypergames_equilibrium_analysis_2018_Aljefri_cybernetics, HG_attacker_defender_2022_Wan, Christo_2024_hypergames_semantic}, but their use in  \glspl{mcs} to obtain task proposal and task assignment strategies remains largely unexplored.

\vspace{-0.5em}
\subsection{Contributions}
\label{subsec:contributions}
%The main contributions of this work are:
The main contribution of this paper is a novel fully decentralized task proposal strategy for the MCSPs and task acceptance strategy for the MUs in the considered \gls{mcs} system.
To determine these strategies, every MCSP and MU requires knowledge about the other decision making agents.
Under realistic settings, we assume that such information is unavailable to both entities and formulate the problem as a dynamic hypergame.
Within the framework of the hypergame, the unknown preferences of the agents are modeled as perceptions and are updated over time.
To account for the incomplete information about the MCS system, we utilize reinforcement learning and propose a learning-aided hypergame solution termed as PACMAB.
The goal of PACMAB is to find strategies for the MCSPs as well as the MUs which maximize their individual revenues.
Towards this goal we make the following key contributions:
\begin{itemize}
    \item 
    To solve the dynamic hypergame, we develop a perception-aware matching solution which accounts for the evolving perceptions of the \glspl{mcsp} to obtain task proposal strategies that maximize their net revenue.
    In this approach, all \glspl{mcsp} know their own preferences over \glspl{mu} and tasks, while the \glspl{mu} know their task efforts in advance.
    Under these idealized assumptions, we can determine a performance upper bound which can be used to compare the performance of the proposed PACMAB algorithm.
    \item We prove that the perception-aware matching solution converges to a stable hyper Nash equilibrium indicating that every \gls{mcsp} achieves a stable assignment in its own subjective game under the presence of perceptions.
    We show that as the perception error reduces, the \glspl{mcsp}' net revenue increases.
    \item 
    To handle the fully unknown preferences at both \glspl{mcsp} and \glspl{mu}, we propose \textit{PACMAB}, a novel, fully decentralized, perception-aware two-sided learning solution.
    PACMAB consists of two components: (a) at every \gls{mcsp}, we implement an upper confidence bound (UCB) multi-armed bandit algorithm to determine an adaptive task proposal strategy without strict requirement of the knowledge of own and other \glspl{mcsp}' preferences, and (b) at every \gls{mu}, we implement a low complexity multi-armed bandit algorithm to learn the task execution efforts required to find an efficient task acceptance strategy.
    PACMAB enables each \gls{mcsp} to maintain and update perceptions about the preferences of other \glspl{mcsp} to obtain its own task proposal strategy.
    \item We analyze the computational complexity of PACMAB from the perspectives of the \glspl{mcsp} and the \glspl{mu} separately. Specifically, we show that from the perspective of the \glspl{mcsp}, the complexity grows only linearly with the number of \glspl{mu}, tasks, and the payment levels used by the \gls{mcsp}.
    From the perspective of the \glspl{mu}, the complexity grows only linearly with the number of \glspl{mcsp}.
    \item 
    We conduct extensive numerical evaluations to compare PACMAB's performance against the benchmark algorithms.
    The results demonstrate superior performance of PACMAB in terms of social welfare, task completion ratio, and sensing quality without the strict requirement of complete information.
\end{itemize}

The rest of the paper is organized as follows.
In Section \ref{sec:system_model}, the considered \gls{mcs} system model is introduced.
Section \ref{sec:problem_formulation} presents the problem formulation as an \gls{mwc} and the reformulation as a dynamic hypergame.
Section \ref{sec:solution_algorithms} provides a solution to this hypergame under information assumptions.
Our proposed PACMAB algorithm is described in Section \ref{sec:learning_aided_hypergame_solution}.
The numerical evaluations are given in Section \ref{sec:numerical_evaluation}.
Finally, conclusions are drawn in Section~\ref{sec:conclusion}.
\section{System model}
\label{sec:system_model}
\begin{figure}
    \centering
    \includegraphics[width=\linewidth]{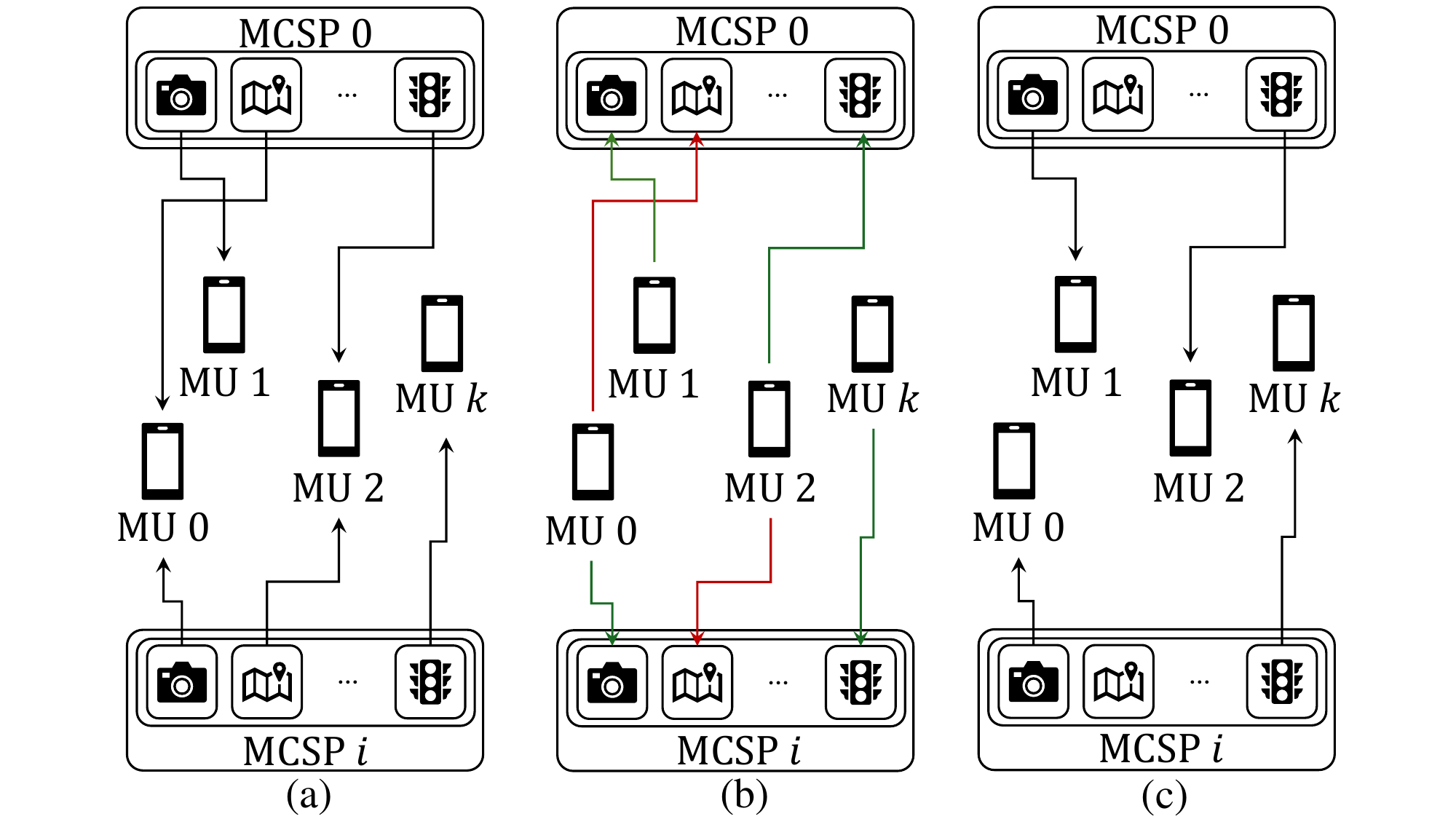}
    \caption{Overview of the system model: (a) \glspl{mcsp} send task offers to \glspl{mu}. (b) \glspl{mu} respond with either an accept or a reject. (c) Assigned \glspl{mu} perform the task, transmit the result back to the respective \gls{mcsp} and receive the payment.}
    \label{fig:system_model}
\end{figure}
\subsection{Overview}
Figure \ref{fig:system_model} illustrates our \gls{mcs} system model.
We consider a set $\mathcal{I}=\{i\}_{i=1}^I$ of $I$ \glspl{mcsp}.
Similarly we assume a set $\mathcal{K}=\{k\}_{k=1}^K$ of $K$ \glspl{mu}.
The time is divided into $T$ discrete time steps and each time step is given by $t\in\{0, 1, \ldots, T-1\}$.
Table~\ref{tab:notations} provides the summary of notations used in this work.

Every \gls{mcsp} $i$ offers $N_i$ different tasks in every time step $t$.
These tasks are collected in the set $\mathcal{O}^i_t$ and are indexed by $n$.
We consider that different types of tasks are present in $\mathcal{O}^i_t$, e.g., temperature sensing, noise level monitoring, or taking a picture or a video.
Different task types are collected in a set $\mathcal{Z}=\{z\}_{z=1}^Z$ of $Z$ task types.
Each available task is given by $O^i_{n,t}$ and is mapped into a task type $z\in\mathcal{Z}$ using a mapping function $g:\mathcal{A}^i_t \rightarrow \mathcal{Z}$ such that $g(O^i_{n,t})=z$.
In each time step, tasks of the same type $z$ are collected in a set $\mathcal{O}^i_{z,t}\subseteq\mathcal{O}^i_t$.
Each task of type $z$ is characterized by the average data size $d_z$ of the raw sensing data which is measured in bits, the task processing complexity $c_z$ measured in CPU cycles/bit, and the average size $s_z$ of the processed sensing result, measured in bits, and  which the \glspl{mu} transmit back to the \gls{mcsp}.
All the tasks in $\mathcal{O}^i_{z,t}$ have identical characteristics.
We assume that only one $\mathrm{MU}$ is required to successfully complete a task $O^i_{n,t}$.
If an \gls{mcsp} requires more sensing results for the same task type $z$, then it can generate more tasks of the same type and propose them to different \glspl{mu}.
The number $\rho^i_z$ of tasks available in one time step, which belong to the same task type $z$, is termed the quota of task type $z$, with $|\mathcal{O}^i_{z,t}|=\rho^i_z$.

\vspace{-3mm}
\subsection{Mobile crowdsensing platforms}
\label{subsec:mcsps}
At the beginning of each time step $t$, every \gls{mcsp} $i$ offers a task $O^i_{n,t}$ of type $z$ to \gls{mu} $k$ along with a payment offer $P^i_{k,n,t} \in \mathcal{P}^i_z$.
$\mathcal{P}^i_z$ is a discrete set containing the possible payments that each \gls{mcsp} can offer for each task type $z$.
The offer is denoted by $\Hat{O}^i_{k,n,t}=\langle O^i_{n,t}, P^i_{k,n,t} \rangle$ and aims at maximizing the net revenue of the \gls{mcsp}.
The \gls{mcsp} transmits this offer to the \gls{mu} and waits for its acceptance or rejection decision.
Once accepted, the \gls{mu} $k$ completes the task $O^i_{n,t}$ and transmits the sensing result $r_{k,n,t}$ back to the \gls{mcsp} $i$ over a wireless channel.
Note that an \gls{mcsp} can offer only one task to an \gls{mu} in time step $t$, as the \glspl{mu} can execute only one task per time step.
If \gls{mu} $k$ accepts task offer $\Hat{O}^i_{k,n,t}$ from \gls{mcsp} $i$, this assignment is denoted by $y^i_{k,n,t}\in \{0,1\}$.
All the assignment decisions for \gls{mcsp} $i$ are stored in the matrix $\textbf{Y}^i_t \in \{0,1\}^{K\times N_i}$.

\gls{mcsp} $i$ and the interested \gls{dr} make a contractual agreement in which the \gls{dr} pays at least $w^i_z$ monetary units to the \gls{mcsp} for every completed task of type $z$.
In addition to the basic payment $w^i_z$, the \gls{dr} is willing to pay more proportional to the quality of the sensing task result.
When \gls{mu} $k$ completes a task, it sends the sensing result $r_{k,n,t}$ back to the \gls{mcsp}.
\gls{mcsp} $i$ calculates the quality factor $q^i_{k,n,t}\in[0,1]$ of the sensing result $r_{k,n,t}$  using a quality function $Q^i_z(\cdot)$ given by
\begin{equation}
    q^i_{k,n,t} = Q^i_z(r_{k,n,t}).
\end{equation}
For a fixed \gls{mu} $k$ and task $O^i_{n,t}$ of type $z$, the quality $q^i_{k,n,t}$ is an unknown random variable which can be observed only after completing the task.
Each \gls{mcsp} evaluates the quality independently and based on the specific aspect of the sensing result it wants to focus on. 
Some examples of the quality functions are the Peak Signal-to-Noise Ratio (PSNR) of an image, or the accuracy and resolution of a temperature measurement.
As $Q^i_z(\cdot)$ is different for each \gls{mcsp}, the same sensing result $r_{k,n,t}$ can have different quality factors depending on to which \gls{mcsp} it is submitted.
The calculation of $w^i_{k,z}$ given by
\begin{equation}
    w^i_{k,z} = (1 + q^i_{k,n,t})w^i_z.
\end{equation}
We emphasize that the proposed system model is flexible and allows alternative functions to be easily integrated for evaluating $w^i_{k,z}$.
As the quality $q^i_{k,n,t}$ of \gls{mu} $k$ is unknown, \gls{mcsp} $i$ does not know $w^i_{k,z}$ in advance.
The utility $U^{\mathrm{MCSP},i}_{k,n,t}$ of the \gls{mcsp} when the task $O^i_{n,t}$ of type $z$ is successfully performed by \gls{mu} $k$ is given by
\begin{equation}
\label{eq:MCSP_utility_per_offer}
    U^{\mathrm{MCSP},i}_{k,n,t} = w^i_{k,z} - P^i_{k,n,t}.
\end{equation}
Thus, the \gls{mcsp} has to balance the quality of the \glspl{mu} with the payments they will receive.
The total utility achieved by \gls{mcsp} $i$ in time step $t$ is given by
\begin{equation}
\label{eq:total_MCSP_utility}
    U^{\mathrm{MCSP}, i}_{t} = \sum_{k=1}^K \sum_{n=1}^{N_i} y^i_{k,n,t}U^{\mathrm{MCSP}, i}_{k,n,t}.
\end{equation}
Since $w^i_{k,z}$ is not known, the \gls{mcsp} estimates its utility while assigning the tasks.
The estimated utility $\bar{U}^{\mathrm{MCSP},i}_{k,n,t}$ for a task of type $z$ assigned to \gls{mu} $k$ is given by
\begin{equation}
\label{eq:estimated_MCSP_utility_per_offer}
    \bar{U}^{\mathrm{MCSP},i}_{k,n,t} = \mathbb{E}\{{U}^{\mathrm{MCSP},i}_{k,n,t}\} = \mathbb{E}\{w^i_{k,z}\} - P^i_{k,n,t}.
\end{equation}
The total estimated utility $\Bar{U}^{\mathrm{MCSP}, i}_{t}$ in $t$ is given by
\begin{equation}
    \Bar{U}^{\mathrm{MCSP}, i}_{t} = \sum_{k=1}^K \sum_{n=1}^{N_i} \Bar{U}^{\mathrm{MCSP}, i}_{k,n,t}.
\end{equation}
% Preamble:
% In the text (right before the table, or in notation section):
% (Define a short phrase used in the table)
% For entries depending on MU k and task O^i_{n,t}, we use the suffix “(k,O^i_{n,t})”.

\begin{table*}[t]
\vspace{-3mm}
\centering
\scriptsize
\setlength{\tabcolsep}{4pt}
\renewcommand{\arraystretch}{1.1}
\caption{Table of notations}
\label{tab:notations}

\begin{tabularx}{\textwidth}{@{} X l X l @{}}
\toprule
\textbf{Description} & \textbf{Notation} & \textbf{Description} & \textbf{Notation} \\
\midrule

Set of \glspl{mcsp}, total available \glspl{mcsp}, \gls{mcsp} index
& $\mathcal{I}$, $I$, $i$
& Quality of \gls{mu} $k$ for task $O^i_{n,t}$
& $q^i_{k,n,t}$ \\

Set of \glspl{mu}, total available \glspl{mu}, \gls{mu} index
& $\mathcal{K}$, $K$, $k$
& Utility of \gls{mcsp} $i$ if MU $k$ completes task $O^i_{n,t}$
& $U^{\mathrm{MCSP},i}_{k,n,t}$ \\

Time horizon, time step index
& $T$, $t$
& Utility of \gls{mu} $k$ after performing task $O^i_{n,t}$
& $U^\mathrm{MU}_{k,n,t}$ \\

Set of task types, total available task types, task type index
& $\mathcal{Z}$, $Z$, $z$
& Task completion time of \gls{mu} $k$ for task $O^i_{n,t}$
& $\tau_{k,n,t}$ \\

Set of payments offered by \gls{mcsp} $i$ per task type $z$, payment index, offered payment to \gls{mu} $k$
& $\mathcal{P}^i_z$, $p$, $P^i_{k,n,t}$
& Task quotas of \gls{mcsp} $i$ for task type $z$
& $\rho^i_z$\\

Set of tasks from \gls{mcsp} $i$ at time step $t$
& $\mathcal{O}^i_t$
& Task completion energy of \gls{mu} $k$ for task $O^i_{n,t}$
& $E_{k,n,t}$ \\

Available task of \gls{mcsp} $i$
& $O^i_{n,t}$
& Sensing time/energy of \gls{mu} $k$ for task $O^i_{n,t}$
& $\tau^\mathrm{sense}_{k,n,t}$, $E^\mathrm{sense}_{k,n,t}$ \\

Sensing data size, complexity, task size of type $z$
& $s_z$, $c_z$, $d_z$
& Computing time/energy of \gls{mu} $k$ for task $O^i_{n,t}$
& $\tau^\mathrm{comp}_{k,n,t}$, $E^\mathrm{comp}_{k,n,t}$ \\

Minimum earning of \gls{mcsp} $i$ from completion of a task $O^i_{n,t}$ of type $z$
& $w^i_z$
& Communication time/energy of \gls{mu} $k$ for task $O^i_{n,t}$
& $\tau^\mathrm{comm}_{k,n,t}$, $E^\mathrm{comm}_{k,n,t}$ \\

Actual earning of \gls{mcsp} $i$ from completion of a task $O^i_{n,t}$ of type $z$ by \gls{mu} $k$
& $w^i_{k,z}$
& \gls{mcsp} $i$’s offer to \gls{mu} $k$ for task $O^i_{n,t}$
& $\hat{O}^i_{k,n,t}=\langle O^i_{n,t}, P^i_{k,n,t}\rangle$ \\

Offered payment from \gls{mcsp} $i$ to \gls{mu} $k$
& $P^i_{k,n,t}$
& Task assignments of \gls{mcsp} $i$
& $y^i_{t}$ \\

Preferences of \gls{mcsp} $i$ over \gls{mu} $k$ and task type $z$
& $S^i_{k,z}$
& Preferences of the other \glspl{mcsp} over \gls{mu} $k$ and task type $z$
& $S^{-i}_{k,z}$ \\

\bottomrule
\end{tabularx}
\end{table*}

\vspace{-4mm}
\subsection{Mobile Units}
\label{subsec:mobileUnits}
In every time step $t$, the \glspl{mu} receive task offers from the \glspl{mcsp}.
Without loss of generality, we assume that all the available \glspl{mu} are capable of performing tasks of all types.
However, the quality of the sensing result may vary depending on which \gls{mu} performs the task.
Every \gls{mu} $k$ may receive multiple task offers from different \glspl{mcsp}.
However, in one time step $t$, \gls{mu} $k$ can perform only one task and thus, it has to decide which task offer to accept and which to reject.
\gls{mu} $k$ makes this decision depending on the efforts required to complete the offered task.
To successfully complete a task $O^i_{n,t}$, the \gls{mu} has to spend time and energy.
Specifically, an \gls{mu} $k$ requires time $\sensingTime$, measured in seconds, to sense and generate sensing data $d_z$, measured in bits.
$\sensingTime$ is drawn from a stationary random distribution with probability distribution function (PDF) $f^z_{\sensingTime}(\sensingTime)$ with expected value $\expectedSensingTime=\mathbb{E}\{\sensingTime\}$.
This expected value depends on the task type $z$ and the capabilities of \gls{mu} $k$.
After generating the raw sensing data $d_z$, \glspl{mu} $k$  has to process it such that the result can be transmitted to the \gls{mcsp} over a wireless channel \cite{C1:own:Dongare2024b, simon2023decentralized}.
The computing time $\computingTime$ for processing the sensing data $d_z$ is given by
\begin{equation}
    \computingTime = \frac{c_z d_z}{f^\mathrm{local}_k},
\end{equation}
where $f_k^\mathrm{local}$ is the CPU frequency of \gls{mu} $k$, measured in Hz.
After processing, the sensing result $r_{k,n,t}$ has size $s_z < d_z$.
This sensing result is then transmitted to \gls{mcsp} $i$.
The transmission time $\communicationTime$ required for this is drawn from a stationary random distribution with PDF $f^z_{\communicationTime}(\communicationTime)$.
The expected value of this distribution is denoted by $\expectedCommunicationTime$ and depends on $s_z$ and the quality of the communication channel between \gls{mu} $k$ and \gls{mcsp} $i$.
We assume that the transmission happens via orthogonal frequency division multiple access (OFDMA), meaning, each \gls{mu} is assigned a communication bandwidth which is orthogonal to that of the other \glspl{mu}.
The total time required for successfully completing the task will be $\totalTime=\sensingTime+\computingTime+\communicationTime$.

Similar to time efforts, the \glspl{mu} also spend energy.
The total energy effort invested by \gls{mu} $k$ to perform task $O^i_{n,t}$ is denoted by $\totalEnergy$ and is given by
\begin{align}
    \totalEnergy = &\sensingEnergy + \computingEnergy + \communicationEnergy \\
     = & \sensingTime  \sensePower + \computingTime  \compPower + \communicationTime  \txPower.
\end{align}
Here $\sensePower, \compPower, \txPower$ represent the sensing, computing, and communication power required by \gls{mu} $k$.
Considering the time and energy efforts, we define the \gls{mu}-specific cost function $\costFunction$ \cite{simon2023decentralized, C1:own:Dongare2024b} as
\begin{equation}
    \costFunction = \alpha_k \totalTime + \beta_k \totalEnergy.
\end{equation}
Note that the cost of performing any task of type $z$ is identical for a given \gls{mu} $k$ irrespective of the \gls{mcsp} $i$ offering it.
This is because the cost depends only on the \gls{mu} capabilities and the task type.
The cost function balances completion time and consumed energy using the \gls{mu}-specific time-cost parameter $\alpha_k$ measured in monetary units per second and energy-cost parameter $\beta_k$ measured in monetary units per joules.
Each \gls{mu} uses $C^\mathrm{effort}_{k,n,t}$ as the minimum payment required to compensate its efforts.
In reality, the \glspl{mu} prefer payments higher than $C^\mathrm{effort}_{k,n,t}$ in order to make profit.
The \gls{mu} profit is calculated as
\begin{equation}
    U^{\mathrm{MU}}_{k,n,t} = y^i_{k,n,t}(P^i_{k,n,t} - C^\mathrm{effort}_{k,n,t}).
\end{equation}
Since the true task efforts $C^\mathrm{effort}_{k,n,t}$ are not known to \gls{mu} $k$ in advance, it estimates its utility as
\begin{equation}
\begin{split}
    \Bar{U}^\mathrm{MU}_{k,n,t}  & = \mathbb{E}\{U^\mathrm{MU}_{k,n,t}|O^i_{n,t}\in\mathcal{O}^i_{z,t}\}\\
                 & = P^i_{k,n,t} - \mathbb{E}\{C^\mathrm{effort}_{k,n,t}\}.
\end{split}
\end{equation}
The \glspl{mu} accept the task offers that maximize their expected estimated utility $\Bar{U}^\mathrm{MU}_{k,n,t}$. They make this decision independently and selfishly.
We assume that the \glspl{mu} provide a feedback to the \glspl{mcsp} when a task is rejected.
The feedback involves which task they have accepted and at what payment.
The \glspl{mu} share this information  as an incentive to receive better and more attractive offers in the future.
Additionally, it helps the \gls{mcsp} to estimate the preferences of \glspl{mu} as well as the strategy of the other \glspl{mcsp}.

\section{Problem formulation}
\label{sec:problem_formulation}
\subsection{Problem formulation as a matching game with contracts}
To capture the fact that both the \glspl{mcsp} and the \glspl{mu} make independent and selfish decisions based on their own preferences, we formulate the task proposal and task acceptance problem using \gls{mwc} \cite{Milgrom_Hatfield_MWC_2005}.
All the \glspl{mcsp} and the \glspl{mu} are considered to be rational and selfish decision makers which aim to maximize their own utilities.
The main goal of \gls{mwc} is to identify a stable matching solution, i.e., task assignments which neither the \glspl{mcsp} nor the \glspl{mu} can improve by changing the assignments.
The \gls{mwc} is a model designed for two-sided matching markets such as our \gls{mcs} system where the \glspl{mcsp} have certain sensing demands from a target area and the \glspl{mu} offer their sensing resources in exchange for payments.
The sensing demands are defined as sensing tasks which the \glspl{mcsp} offer to the \glspl{mu} along with a certain payment as an incentive.
The matching game $\mathcal{G}$ in time step $t$ is formally defined by the tuple $\mathcal{G}=\{\mathcal{I}, \mathcal{K}, \mathcal{O}^i_t, \MUpreference, \MCSPTaskpreference\}$ where $\MUpreference$ represents the preference ordering of \gls{mu} $k$, and, similarly, $\MCSPTaskpreference$ represents the preference ordering of \gls{mcsp} $i$.
We also define a contract $x^i_{k,z,p} = \{i,k,O^i_{n,t}, P^i_{k,n,t}\}$ for the \gls{mwc} \cite{Milgrom_Hatfield_MWC_2005}.
Each contract $x^i_{k,z,p}$ is bilateral, i.e., it is associated to one \gls{mcsp} and one \gls{mu}.
The finite set $\mathcal{X}$ contains all possible contracts.

The \glspl{mcsp}' preference ordering $\MCSPTaskpreference$ ranks the contracts in $\mathcal{X}$ in decreasing order of the expected utility, i.e.,
\begin{equation}
\begin{split}
    \langle k, z, P^i_{k,n,t}\rangle \MCSPTaskpreference \langle l,z',P^i_{l,n',t}\rangle &\\
    \iff \Bar{U}^{\mathrm{MCSP},i}_{k,n,t} \MCSPTaskpreference \Bar{U}^{\mathrm{MCSP},i}_{l,n',t}.
\end{split}
\end{equation}
In other words, \gls{mcsp} $i$ prefers to offer task $O^i_{n,t}$ of type $z$ to \gls{mu} $k$ for the payment $P^i_{k,n,t}$ more than it prefers to offer task $O^i_{n',t}$ of type $z'$ to \gls{mu} $l$ at the payment $P^i_{l,n',t}$.
This is because the former offer yields higher expected utility.
Note that tasks of same type will also yield different expected utilities depending on which \gls{mu} performs the task and the payment offered.
Similarly, the \glspl{mu}' preference ordering $\MUpreference$ ranks the received task offers from different platforms according to their expected utility, i.e.,
\begin{equation}
    \Hat{O}^i_{k,n,t} \MUpreference \Hat{O}^{j}_{k,n',t} \iff \Bar{U}^\mathrm{MU}_{k,n,t} \MUpreference \Bar{U}^\mathrm{MU}_{k,n',t}.
\end{equation}
The \gls{mu} will always choose the offer that maximizes its expected utility $\Bar{U}^\mathrm{MU}_k$.
The binary variable $y^i_{k,n,t}=1$ if \gls{mu} $k$ accepts task offer $\Hat{O}^i_{k,n,t}$, and $y^i_{k,n,t}=0$ otherwise.
We define a stable task assignment as follows:
\begin{definition}
    A task assignment $\mathbf{Y}_t$ is unstable if there are two \glspl{mu}, \gls{mu} $k$ and \gls{mu} $l$, and two task offers, $\Hat{O}^i_{k,n,t}$ and $\Hat{O}^{j}_{l,n',t}$, from \glspl{mcsp} $i$ and $j$ such that:\\
    a) $y^i_{k,n,t}=1$, i.e., \gls{mu} $k$ has accepted the task offer $\Hat{O}^i_{k,n,t}$.\\
    b) $y^{j}_{l,n',t}=1$, i.e., \gls{mu} $l$ has accepted the task offer $\Hat{O}^{j}_{k,n',t}$.\\
    c) $\langle l, z', P^i_{l,n',t}\rangle \MCSPTaskpreference \langle k,z,P^i_{k,n,t}\rangle$ and $\Hat{O}^{j}_{k,n',t} \MUpreference \Hat{O}^{i}_{k,n,t}$, i.e., \gls{mcsp} $i$ would prefer contract $\langle l, z', P^i_{l,n',t}\rangle$ instead of the assigned contract $\langle k,z,P^i_{k,n,t}\rangle$ and \gls{mu} $k$ would also prefer task offer $\Hat{O}^{j}_{k,n',t}$ over the assigned offer $\Hat{O}^{i}_{k,n,t}$.
\end{definition}
Consequently, a stable task assignment solution maximizes the achieved utilities of the \glspl{mcsp} and the \glspl{mu}, where neither of them can unilaterally change their task assignment strategy to improve their utility.

\subsection{Problem reformulation as a dynamic hypergame}
\label{subsec:hypergame_formulation}
To solve game $\mathcal{G}$, every \gls{mcsp} requires its own preference ordering over all the possible contracts and the task assignment strategy of all other \glspl{mcsp}.
In reality, the \glspl{mcsp} do not know about the strategies of the other \glspl{mcsp}.
Therefore, in order to obtain a task assignment strategy, every \gls{mcsp} maintains its own perception about the other \glspl{mcsp}' preferences.
As a result, the utility achieved by the \gls{mcsp} from task assignment strategy depends on the accuracy of its perceptions.
The lack of knowledge about the other \glspl{mcsp} motivates us to reformulate the \gls{mwc} into a dynamic level-one hypergame in which misperceptions about the preferences of other \glspl{mcsp} exist.
The \glspl{mcsp} observe the outcomes of the repeated matching game to improve their perceptions.
Formally, a dynamic level-one hypergame $\mathcal{H}_t$ \cite{BENNETT1977749, BENNETT1980489, BENNETT1980293Bidders_and_dispensers, Kovach_Gibson_Lamont_HG_2015} is defined as follows.
\begin{definition}
    A \emph{hypergame} $\mathcal{H}_t$ is given by $(\mathcal{I}, (\mathcal{G}^i)_{i\in \mathcal{I}})$, where $\mathcal{I}$ is a set of $I$ \glspl{mcsp} and $\mathcal{G}^i_t=(\mathcal{I}, S^i, U^{\mathrm{MCSP},i}_t)$ is a subjective game of the $i^\text{th}$ \gls{mcsp}, where:\\
    a) $\mathcal{I}$ is a set of \glspl{mcsp} perceived by \gls{mcsp} $i$.\\
    b) $S^i=\times_{j\in \mathcal{I}}S^i_j$ is a set of strategies perceived by \gls{mcsp} $i$, where $S^i_j$ is the set of the strategies of \gls{mcsp} $j$ perceived by \gls{mcsp} $i$.\\
    c) $U^{\mathrm{MCSP},i}_t$ is the utility function of \gls{mcsp} $i$.
\end{definition}
Note that the set $\mathcal{I}$ contains all the \glspl{mcsp} in the \gls{mcs} system including the \gls{mcsp} $i$ itself.
Due to the existence of the \gls{mcsp}-specific perceptions, every \gls{mcsp} virtually plays its subjective game $\mathcal{G}^i_t$ in which it uses its own perceptions and own preferences to device a task proposal strategy.
In hypergame $\mathcal{H}_t$, the strategy $S^i$ of \gls{mcsp} $i$ depends on its perception about the other \glspl{mcsp}, denoted by $-i$.
If the strategy $S^i_j$ that \gls{mcsp} $i$ perceives about \gls{mcsp} $j$, with $j \neq i$, differs from \gls{mcsp} $j$'s actual strategy, then \gls{mcsp} $i$ \textit{misperceives} \gls{mcsp} $j$.
Misperceptions degrade the achieved utilities of the \glspl{mcsp} and the achieved utilities of the \glspl{mu} as they result in suboptimal task assignments.
The definition of rationality in a hypergame remains subjective to the perceived game $\mathcal{G}^i_t$ of \gls{mcsp} $i$.
If \gls{mcsp} $i$ misperceives other \glspl{mcsp}, the obtained task assignment strategy of \gls{mcsp} $i$ may not be rational to other players in their perceived game $\mathcal{G}^{-i}_t$ and also in the base game $\mathcal{G}$.
However, the obtained strategy can be rational for \gls{mcsp} $i$ in its own perceived game $\mathcal{G}^i_t$ if it is the best response to its perceptions.
The repeated nature of the \gls{mwc} allows the \glspl{mcsp} to update their perceptions using the outcomes of the game as feedback.
The selected strategy of \gls{mcsp} $i$ may result in unexpected outcomes or surprises due to the existence of misperceptions.
As a result, the \glspl{mcsp} have an intrinsic motivation to use this feedback to update their perceptions and adjust their task proposal strategy.
Over time, the misperceptions decrease as the estimate of the other \glspl{mcsp}' strategy becomes more accurate.
When the perception error between the perceived preferences of other \glspl{mcsp} and their respective true preferences reaches a constant value, the solution becomes a stable hypergame solution \cite{Sasaki_2008}.
After this point, \gls{mcsp} $i$ has no intrinsic motivation to update its perceptions because there are no matching surprises in the outcome.
The hypergame $\mathcal{H}_t$ is then said to have achieved stability.
In our MCS system, since the base game is an \gls{mwc}, stability means that the decision making entities have found an assignment from which they cannot deviate unilaterally without reducing their individual achieved utilities.
In this context, a stable hypergame outcome indicates that the decision making agents have obtained an assignment under their own perceptions from which they will not deviate.
\vspace{-4mm}
\section{Perception-aware Matching Solution}
\subsection{Overview}
\label{sec:solution_algorithms}
\begin{figure}[t]
    \centering
    \includegraphics[width=\linewidth]{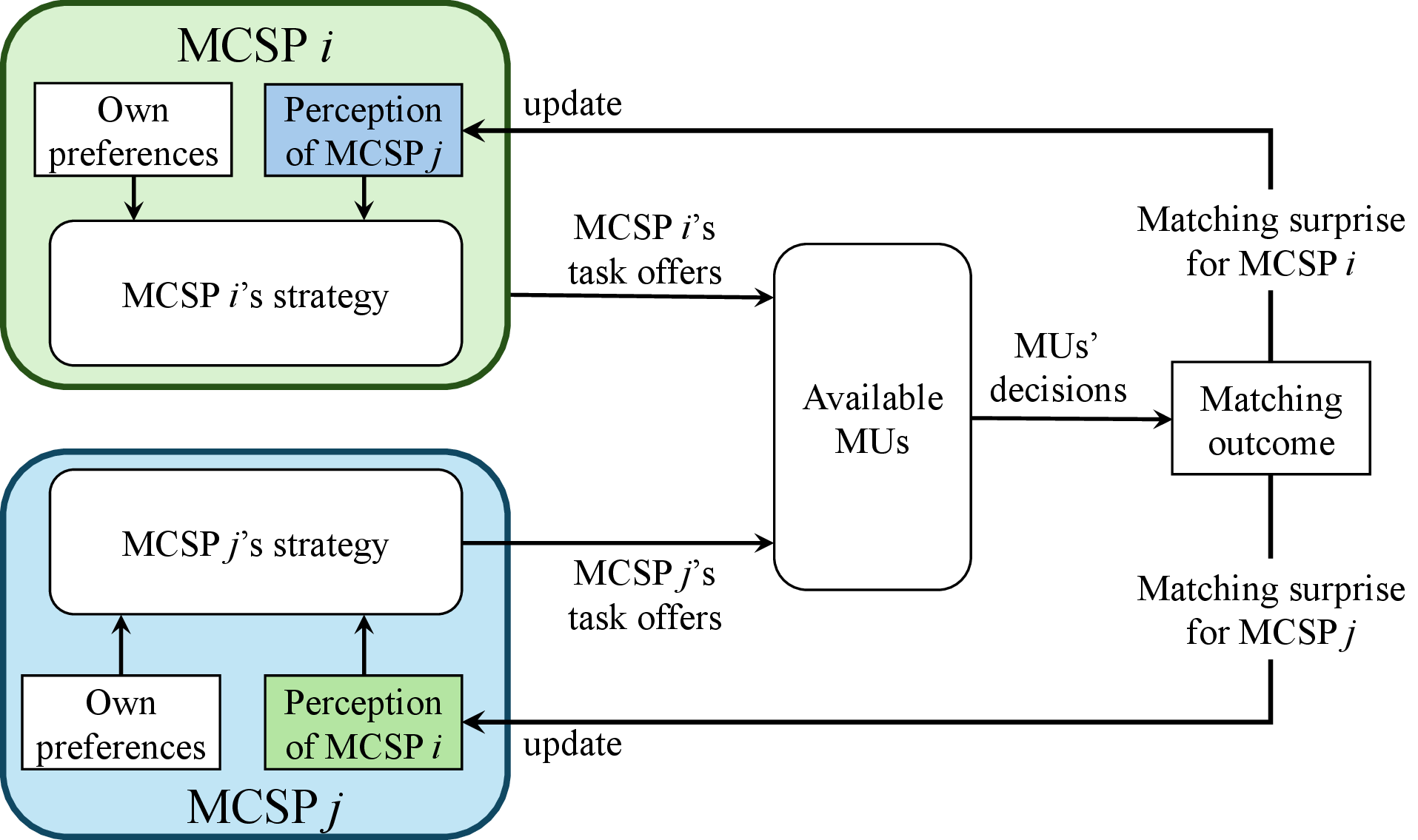}
    \caption{Illustration of interaction between $2$ \glspl{mcsp} in a perception-aware matching solution}
    \label{fig:prism}
\end{figure}

To optimally solve the hypergame $\mathcal{H}_t$ under the presence of \gls{mcsp}-specific perceptions, the \glspl{mcsp} must know their own preferences over all \gls{mu}-task combinations.
Similarly, the \glspl{mu} must know their own preferences over the task types.
The \glspl{mcsp} maintain their own perceptions about the preferences of the other \glspl{mcsp} and update them over repeated interactions.
Using this information, every \gls{mcsp} obtains a task proposal strategy that maximizes its own utility while accounting for the competition.
Similarly, the \glspl{mu} accept the offers that maximize their own utilities.
The success of the task proposal and task acceptance strategies resulting from the solution of $\mathcal{H}_t$ highly depends on the accuracy of the perceptions.
\begin{algorithm}[t]
    \caption{MCSP $i$'s perception-aware matching algorithm}\label{alg:proposed_algorithm_PerfectHG}
    \begin{algorithmic}[1]
    \footnotesize
        \REQUIRE $\theta^i_i$.
        \STATE \textbf{Initialization:}
        \STATE Initialize $\theta^{i}_{-i,k,z}=0$, $P^{i,\mathrm{min}}_{k,n,t}=0$.
        \FOR{$t=1, 2,\ldots,T$}
            \STATE Initialize $\mathcal{O}^i_t$
            \STATE Draw a random variable $v$ from $\mathcal{U}[0, 1]$.
            \IF{$\epsilon_t > v$}
                \FOR{$k = 1, 2, \ldots, K$}
                    \STATE Choose $O^i_{n,t} \leftarrow z$ from a uniform random distribution of available task type $z \in \mathcal{O}^i_{z,t}$.\hfill \algorithmiccomment{Exploration}
                    \STATE Draw random variable $v_a$ from $\mathcal{U}[0,1]$.
                    \IF{$v_a < 0.5$}
                        \STATE $P^i_{k,n,t} \leftarrow P^{i,\mathrm{min}}_{k,n,t}$
                    \ELSE
                        \STATE $P^i_{k,n,t} \leftarrow \theta^i_{i,k,z}$
                    \ENDIF
                    \STATE Send sensing offer $\hat{O}^i_{k,n,t} = \langle O^i_{n,t}, P^i_{k,n,t} \rangle$
                \ENDFOR
            \ELSE
                \STATE Get $k, O^i_{n,t} \quad \forall k, n$. \hfill
                \algorithmiccomment{Exploitation}
                \STATE Calculate shadow price $\Bar{P}^i_{k,z}$ by removing $z$ from $\theta^i_{-i}$.
                \STATE Set $P^i_{k,n,t} = \max(P^{i,\mathrm{min}}_{k,n,t}, \Bar{P}^i_{k,z})$.
                \STATE Send sensing offer $\hat{O}^i_{k,n,t} = \langle O^i_{n,t}, P^i_{k,n,t} \rangle$.
            \ENDIF
            \STATE Get MU's response to the task offer.
            \FOR{\textbf{each} offer $O^i_{k,n,t}$}
                \IF{$O^i_{k,n,t}$ is \textit{accepted}, i.e., $y^i_{k,n,t}=1$}
                    \STATE MU $k$ performs the task, \gls{mcsp} $i$ observes $U^{\mathrm{MCSP},i}_{k,n,t}$, $q^i_{k,n,t}$.
                \ELSE
                    \IF{Rejection due to \textit{Negative Utility}}
                        \STATE Update $P^{i,\mathrm{min}}_{k,n,t} \leftarrow P^{i,\mathrm{min}}_{k,n,t} + 1$
                    \ELSE
                        \STATE MU $k$ accepts MCSP $j$'s task offer
                        \STATE Update $\theta^{i}_{j,k,z} = \mathrm{max}\{\theta^{i}_{j,k,z}, P^{j}_{k,n,t}\}$
                    \ENDIF
                \ENDIF
            \ENDFOR
            \STATE Update $\epsilon_t = \epsilon_{t-1} * \epsilon_d$
        \ENDFOR
    \end{algorithmic}
\end{algorithm}

In this section, we present a perception-aware matching algorithm to optimally solve the dynamic hypergame $\mathcal{H}_t$ formulated in Section \ref{sec:problem_formulation}.
Using this algorithm, every \gls{mcsp} obtains a task proposal strategy which maximizes its own utility.
For this algorithm, we assume that every MCSP $i$ knows its own expected revenue $\theta^i_{i,k,z}\triangleq \mathbb{E}\{w^i_{k,z}\}$ for each \gls{mu}-task type pair $(k,z)$.
The assumption implies that every \gls{mcsp} knows the qualities of every \gls{mu} for every task type $z$.
On the \gls{mu}-side, we assume that every \gls{mu} knows the task execution efforts for all task types in advance.
Note that the availability of such information to any of the entities in advance is unrealistic, and makes it impossible to implement the perception-aware matching solution in real-world \gls{mcs} applications.
However, we present this algorithm as a theoretical upper bound solution for the dynamic hypergame.
We also provide the equilibrium analysis of this solution to illustrate that it is a stable hyper Nash equilibrium solution.
The known expected revenues and the possible payments induce preference orderings of \gls{mcsp} $i$ over all \gls{mu}-task combinations.
Note that MCSP $i$ does not know the expected revenues of other \glspl{mcsp} and maintains its perceptions $\{\theta^i_{j,k,z}\}_{\forall j \in \mathcal{I},j\neq i}$ which are updated from observed outcomes or surprises.
To simplify the notation, we use $\theta^i_{-i}$ to denote the perception  \gls{mcsp} $i$ has about all other \glspl{mcsp} for all \glspl{mu} and task types.

The core idea of the algorithm is that every \gls{mcsp} uses its preferences over \glspl{mu} and tasks and its perceptions to devise a task proposal strategy.
This is illustrated in Fig.~\ref{fig:prism}.
At first, the perceptions are inaccurate, and thus the resulting task proposal strategy is also suboptimal.
Over repeated interactions, the \glspl{mcsp} observe the outcomes of the game and use these outcomes to update their perceptions about the other \glspl{mcsp}.
As the perception error, i.e., the difference between the true expected revenues of the other \glspl{mcsp} and the perceived values, reduces, the task proposal strategy improves.
To learn their perceptions, the \glspl{mcsp} employ an $\epsilon$-greedy algorithm which balances between exploring new MU-task combinations at different payments and exploiting their current perceptions to obtain offers which maximize the expected utility using the $\epsilon$ parameter.

\subsection{Solution description}
As shown in Alg.~\ref{alg:proposed_algorithm_PerfectHG}, every \gls{mcsp} $i$ initializes $\theta^i_{-i}$ and the minimum payment matrix $P^{i,\mathrm{min}}_{k,z}$ per \gls{mu} $k$ and task type $z$ with zeros (Line 2).
In every time step $t$, the \gls{mcsp} $i$ obtains the set $\mathcal{O}^i_t$ of tasks it has to assign to the available \glspl{mu} (Line 4).
With probability $\epsilon_t$, the \gls{mcsp} explores or else, exploits (Line 5,6).
In the exploration phase, the \gls{mcsp} probes the environment by offering a random payment between $P^{i,\mathrm{min}}_{k,z}$ and its own expected revenue $\theta^i_{i,k,z}$ with a probability $\epsilon_a$ (Line 10-14).
The \gls{mcsp} creates a task offer $\Hat{O}^i_{k,n,t}=\langle O^i_{n,t}, P^i_{k,n,t}\rangle$ for every \gls{mu} (Line 15).

In the exploitation phase, the \gls{mcsp} $i$ utilizes its perceptions $\theta^i_{-i}$ about the other \glspl{mcsp} about the willingness of the other \glspl{mcsp} to offer tasks to the \glspl{mu}.
\gls{mcsp} $i$ then uses the well-known Hungarian algorithm \cite{Kuhn_hungarian_1955, munkres1957algorithms} to obtain a task proposal strategy from the expected revenue $\theta^i_i$ of itself and its perception $\theta^{i}_{-i}$ about the other \glspl{mcsp} (Line 18).
The Hungarian algorithm essentially identifies task proposals that will maximize the expected utility of the \gls{mcsp}.
Additionally, the Hungarian algorithm also estimates proposals of the other \glspl{mcsp} from the perceptions of \gls{mcsp} $i$.
The next step is to identify suitable payments for the MU-task combinations found using the Hungarian algorithm.
\gls{mcsp} $i$ estimates how valuable a task type $z$ is to its competitor by calculating the decrease in the overall expected utility of the competitor if it had one less task of that type.
From this value, \gls{mcsp} $i$ evaluates how much it needs to offer to outbid the competitor for some \gls{mu} (Line 20).
After sending the sensing offer to the \glspl{mu}, every \gls{mcsp} waits for their responses (Line 21).
If the offer $\Hat{O}^i_{k,n,t}$ is accepted, i.e., $y^i_{k,n,t}=1$, \gls{mu} $k$ performs task $O^i_{n,t}$ (Line 25).
With this assignment, the \gls{mcsp} $i$ achieves $U^{\mathrm{MCSP},i}_{k,n,t}$ and pays $P^i_{k,n,t}$ to the \gls{mu} $k$ (Line 26).
If the offer is rejected, i.e. $y^i_{k,n,t}=0$, the \gls{mu} conveys this decision to the \gls{mcsp} $i$ along with a reason.
In case of a rejection due to a negative utility, the \gls{mcsp} $i$ updates $P^{i,\mathrm{min}}_{k,z}$ by increasing the payment index $p$ by one (Line 28-29).
Otherwise, the rejection is because the \gls{mu} $k$ chose the offer of some other \gls{mcsp} $j$ with $j\neq i$.
This is a matching surprise for the \gls{mcsp} $i$ which then triggers the perception update as $\theta^{i}_{j,k,z} = P^{j}_{k,n,t}$ for the accepted \gls{mcsp} $j$ (Line 32).
Finally, $\epsilon_t$ is updated for the next time step (Line 36).

\subsection{Equilibrium analysis}
\label{subsec:equilibrium_analysis_HNE}
We will now analyze the strategies of the competing \glspl{mcsp} obtained from the perception-aware matching algorithm to solve the hypergame $\mathcal{H}_t$.
We show that the algorithm reaches a \textit{Hyper Nash Equilibrium} (HNE) \cite{Kovach_Gibson_Lamont_HG_2015, Sasaki_2008} under the existence of \gls{mcsp}-specific perceptions.
We collect all the matching decisions made in time step $t$ in matrix $\textbf{Y}_t$.

The \gls{mcsp}-specific perceptions induce \gls{mcsp}-specific perceived expected utility $U^{\mathrm{MCSP},i}_t({Y}^i_t;\theta^i)$, where $\theta^i=[\theta^i_i \theta^i_{-i}]$.
Every \gls{mcsp} $i$ aims to obtain a task proposal strategy which is its best response to the perceptions $\theta^i_{-i}$ about the other \glspl{mcsp}' that maximizes its perceived expected utility.
As discussed in the previous section, due to the existence of perceptions, every \gls{mcsp} plays a game $\mathcal{G}^i_t$ which is its own subjective game due to its own perceptions.
We define a hyper Nash equilibrium subject to the \gls{mcsp}-specific perceptions of all \glspl{mcsp}.
\begin{definition}
    For a fixed perception profile $\theta = (\theta^i)_{i\in\mathcal{I}}$, a joint task assignment $\textbf{Y}^*=(\textbf{Y}^i_t)_{i\in \mathcal{I}}$ is a Hyper Nash Equilibrium (HNE) solution for every \gls{mcsp} $i \in \mathcal{I}$ if, $\textbf{Y}^* \in \arg\max_{\textbf{Y}^i_t \in \textbf{Y}_t} U^{\mathrm{MCSP},i}_t(\textbf{Y}^i_t;\theta^i)$. Equivalently, for all unilateral deviations, $\Tilde{\textbf{Y}}^i_t \in \textbf{Y}_t$, $U^{\mathrm{MCSP},i}_t(\textbf{Y}^i_t;\theta^i) \geq U^{\mathrm{MCSP},i}_t(\Tilde{\textbf{Y}}^i_t;\theta^i)$.
\end{definition}
This means that each \gls{mcsp} $i$'s offers must be its best response to its perceptions about the other \glspl{mcsp} in its perceived game $\mathcal{G}^i_t$.
Due to the misperceptions, even with the best response strategy, \glspl{mcsp} can experience matching surprises in the MCS system.
The matching surprises occur when an \gls{mu} was expected to accept a task offer but rejected or when an \gls{mu} accepted a task offer at a lower payment than before.
The perceptions are updated when a matching surprise is encountered.
The mean absolute error between the true expected revenue of the \gls{mcsp} $j$, i.e., $\theta^{j}_{j}$, and the perception of this parameter $\theta^i_{j}$ maintained by \gls{mcsp} $i$ is called perception error of \gls{mcsp} $i$ given by, $\Delta {\theta^i}(t) = \sum_{j} \mathbb{E}\{|\theta^{j}_{j}-\theta^i_{j}|\} \forall j \in \mathcal{I}$ where $j\neq i$.
In our \gls{mcs} system, $\Delta {\theta^i}(t)$ is a monotonically decreasing function since the perceptions $\theta^i_{-i}$ are monotonically increasing and the true revenue $\theta^{j}_{j} \, \forall j\in \mathcal{I}$ is fixed.
In the context of a perception error, an HNE solution is a stable HNE (SHNE) if the profile of the strategies is a NE in the subjective games of the \glspl{mcsp}, i.e., SHNE$(\mathcal{H}_t)=\cap_{i\in\mathcal{I}}\mathcal{N}(\mathcal{G}^i_t)$ \cite{Sasaki_2008, Kovach_Gibson_Lamont_HG_2015}.
In a repeated matching game with contracts, this means that the perceptions have stabilized to a constant value and no player has an intrinsic motivation to update its perception to improve its expected utility $U^{\mathrm{MCSP},i}_t({Y}^i_t ;\theta^i)$ \cite{Sasaki_2008, Kovach_Gibson_Lamont_HG_2015}.
When the perceptions converge, i.e., $\Delta \theta^i(t) = 0$, each \gls{mcsp} $i$ achieves an optimal best response strategy.
This is an SHNE solution that the \glspl{mcsp} have achieved as SHNE$({\mathcal{H}_t})=\times_{i\in\mathcal{I}}\mathcal{N}(\mathcal{G}^i_t)$.
When the misperceptions disappear from every \gls{mcsp}, all of the \glspl{mcsp} are practically playing the base game $\mathcal{G}$ with complete information which helps them achieve their highest possible utility $U^{*\mathrm{MCSP},i}_{t}$.
At this point, for every \gls{mcsp} $i$, the strategies that other \glspl{mcsp} have chosen are consistent with the \gls{mcsp} $i$'s anticipation and there is no incentive to update their perceptions further, i.e., a stable solution is achieved \cite{Sasaki_2008, Kovach_Gibson_Lamont_HG_2015}.

In the following, we summarize the assumptions made for the perception-aware matching algorithm.
\begin{assumption}\label{assump:rational}
All MCSPs maximize their expected utility:
$$\boldsymbol{Y}^i_t \in \argmax_{\boldsymbol{Y}^i \in \mathcal{Y}_i} U^{\text{MCSP},i}_t(\boldsymbol{Y}^i; \theta^i)$$
\end{assumption}
\begin{assumption}\label{assump:truthful}
When \gls{mu} $k$ rejects offer $\Hat{O}^i_{k,n,t}$ from \gls{mcsp} $i$ because it accepted offer $\Hat{O}^j_{k,n',t}$ from \gls{mcsp} $j \neq i$, the \gls{mu} truthfully reveals the tuple $(j, P^j_{k,n',t}, z)$ to \gls{mcsp} $i$. % with probability 1.
\end{assumption}
\begin{assumption}\label{assump:bounded}
There exists $M < \infty$ such that $\theta^i_{i,k,z} \leq M$ for all $i \in \mathcal{I}$, $k \in \mathcal{K}$, $z \in \mathcal{Z}$.
\end{assumption}
\begin{assumption}\label{assump:exploration}
The exploration rate $\{\epsilon_t\}_{t=0}^{T-1}$ satisfies
$$\sum_{t=0}^{\infty} \epsilon_t = \infty \quad \text{and} \quad \sum_{t=0}^{\infty} \epsilon_t^2 < \infty.$$
\end{assumption}
\begin{theorem}\label{thm:prism_convergence}
Under Assumptions~\ref{assump:rational}-\ref{assump:exploration}, the perception-aware matching algorithm satisfies the following properties:
\begin{enumerate}
    \item \textbf{Monotone Perception Convergence:} The perception error
    $\Delta\theta^i(t) = \sum_{j \neq i} \sum_{k,z} \mathbb{E}[|\theta^j_{j,k,z} - \theta^i_{j,k,z}(t)|]$
    is monotonically non-increasing in $t$.
    
    \item \textbf{Almost Sure Convergence:}
    $\lim_{t \to \infty} \Delta\theta^i(t) = \Delta\theta^i_\infty \quad \text{almost surely}$
    where $\Delta\theta^i_\infty \geq 0$ is the residual perception error.
    
    \item \textbf{Exponential Convergence Rate:} There exist constants $\lambda > 0$ and $\epsilon_{\text{noise}} \geq 0$ such that
    $\mathbb{E}[\Delta\theta^i(t)] \leq \Delta\theta^i(0) \cdot e^{-\lambda t} + \frac{\epsilon_{\text{noise}}}{\lambda}$
    
    \item \textbf{Utility Convergence:} Let $U^{\text{MCSP},i}_{\text{SHNE}}$ denote the utility at the stable hyper Nash equilibrium. Then
    $$\liminf_{t \to \infty} \mathbb{E}[U^{\text{MCSP},i}_t] \geq U^{\text{MCSP},i}_{\text{SHNE}} - L \cdot \Delta\theta^i_\infty$$
    where $L$ is the Lipschitz constant of the utility function.
\end{enumerate}
\end{theorem}
\begin{IEEEproof}
See Appendix. %~\ref{proof_theorem1}.
\end{IEEEproof}

\vspace{-6mm}
\section{Proposed Algorithm}
\label{sec:learning_aided_hypergame_solution}
\subsection{Overview}
\begin{algorithm}[t]
    \caption{MCSP $i$'s learning-based perception-aware matching algorithm}\label{alg:proposed_algorithm}
    \begin{algorithmic}[1]
        \footnotesize
        \STATE \textbf{Initialization:}
        \STATE UCB values for $\{\langle k, z, p \rangle : 0\} \forall k, z, p$, $\gamma^{\mathrm{win},i}_{k,z}=0$, $\gamma^{\mathrm{lost},i}_{k,z}=0$
        \FOR{$t=1, 2,\ldots,T$}
            \STATE $\zeta^\mathrm{a}_t, \rho^{\mathrm{a},i}_{z} = \phi$.
            \STATE Check available tasks $\mathcal{A}^i_t$.
            \STATE Compute feasible set $\mathcal{\Tilde{O}}^{\mathrm{feas},i}_{t}$ based on $\theta^{i}_{-i}$, $\gamma^{\mathrm{win},i}_{k,z}$, and $\gamma^{\mathrm{lost},i}_{k,z}$.
            \FOR{$k=1,\ldots,K$}
                \FOR{$z=1,\ldots,Z$}
                    \IF{$z\in \mathcal{O}^i_{z,t}$ and $(k,z) \in \mathcal{\Tilde{O}}^{\mathrm{feas},i}_{t}$}
                        \FOR{$p=1,\ldots,P$}
                        \STATE Get UCB values in a set $\{\langle k, z, p \rangle : \mathrm{UCB}^i_{k,z,p}\}$.
                        \ENDFOR
                    \ENDIF
                \ENDFOR
            \ENDFOR
            \STATE Sort the set $\{\langle k, z, p \rangle : \mathrm{UCB}^i_{k,z,p}\}$ in decreasing order.
            \FOR{\textbf{each} candidate in $\{\langle k, z, p \rangle : \text{UCB score}\}$}
                \IF{$k \in \zeta^\mathrm{a}_t$}
                    \STATE \textit{continue}
                \ENDIF
                \IF{$\rho^{\mathrm{a},i}_{z} \geq \rho^i_z$}
                    \STATE \textit{continue}
                \ENDIF
                \STATE Create offer $O^i_{k,n,t} = \langle O^i_{n,t}, P^i_{k,n,t} \rangle$.
                \STATE Update $\zeta^\mathrm{a}_t \leftarrow k$.
                \STATE Update $\rho^{\mathrm{a},i}_{z} = \rho^{\mathrm{a},i}_{z} + 1$.
            \ENDFOR
            \STATE Get MU's response to the task offer. \hfill \algorithmiccomment{Algorithm \ref{alg:MUs_acceptance_algorithm}}
            \FOR{\textbf{each} offer $O^i_{k,n,t}$:}
                \IF{$O^i_{k,n,t}$ is \textit{accepted}, i.e., $y_{k,n,t}=1$}
                    \STATE Increase acceptance counter $\gamma^{\mathrm{win},i}_{k,z} = \gamma^{\mathrm{win},i}_{k,z} + 1$
                    \STATE MU $k$ performs the task, \gls{mcsp} $i$ observes $U^{\mathrm{MCSP},i}_{k,n,t}$, $q^i_{k,n,t}$.
                    \STATE Update estimates $\hat{U}^{\mathrm{MCSP}^i}_{k,n,t}$.
                \ELSE
                    \STATE Increase the rejection counter $\gamma^{\mathrm{lost},i}_{k,z} = \gamma^{\mathrm{lost},i}_{k,z} + 1$
                    \STATE $\hat{U}^{\mathrm{MCSP},i}_{k,n,t} \leftarrow \hat{U}^{\mathrm{MCSP},i}_{k,n,t-1}$. \hfill
                    \algorithmiccomment{Eq.~(\ref{eq:pacmab_MCSP_utility_update})}
                \ENDIF
            \ENDFOR
        \ENDFOR
    \end{algorithmic}
\end{algorithm}
The perception-aware matching algorithm presented in the previous section solves the dynamic hypergame $\mathcal{H}_t$ under the assumption that \glspl{mcsp} and \glspl{mu} know their own individual preferences.
In realistic scenarios, the \glspl{mcsp} and the \glspl{mu} do not know their own preferences in advance.
Additionally, the \glspl{mcsp} do not know the strategies of other \glspl{mcsp} which, however they need to obtain an optimal best response strategy.
To help the \glspl{mcsp} and the \glspl{mu} to learn their own preferences based on their expected utilities, we propose a fully decentralized perception-aware combinatorial multi-armed bandit (PACMAB) solution.
This algorithm has two components: (i) \glspl{mcsp}' perception-aware online learning algorithm for task assignment strategy, and (ii) \glspl{mu}' multi-armed bandit based online learning for task offer acceptance strategy.
Given these two components, PACMAB is essentially a multi-agent multi-armed bandit algorithm in which different players have different goals and selfishly and independently decide on their strategy.

For the \gls{mcsp}'s task assignment problem, we employ an upper confidence bound (UCB)-based algorithm.
Since the action space of each \gls{mcsp} is prohibitively large, an algorithm capable of systematically addressing the exploration-exploitation trade-off through uncertainty reduction is essential.
To this end, we incorporate perception-based action pruning, which restricts each \gls{mcsp}'s candidate arm set to those expected to yield superior performance, rather than exhaustively evaluating all the possible arms.
This pruning mechanism significantly accelerates the convergence of the proposed PACMAB algorithm.

\vspace{-3mm}
\subsection{Perception-aware combinatorial multi-armed bandit (PACMAB) solution}
The algorithm is presented in Alg.\ref{alg:proposed_algorithm}.
Every \gls{mcsp} initializes the UCB values with zeros.
Additionally every \gls{mcsp} $i$ maintains an acceptance counter $\gamma^{\mathrm{win},i}_{k,z}$ and a rejection counter $\gamma^{\mathrm{lost},i}_{k,z}$ per \gls{mu} and task type combination (Line 2).
These counters keep track of how many times offers involving the $(k,z)$ combination were accepted and rejected, respectively.
In every time step $t$, a set $\zeta^\mathrm{a}_t$ and a vector $\rho^{\mathrm{a},i}_z$ are initialized with $\phi$ and zeros, respectively (Line 4).
The set $\zeta^\mathrm{a}_t$ monitors assigned \glspl{mu} in the current time step $t$ such that one \gls{mu} will receive only one task from \gls{mcsp} $i$.
The vector $\rho^{\mathrm{a},i}_z$ monitors assigned tasks per task type $z$ such that the task quotas $\rho^i_z$ are respected.
The \gls{mcsp} $i$ then checks the available tasks $\mathcal{O}^i_t$ to be performed (Line 5).
Out of all possible task offers, \gls{mcsp} $i$ evaluates a feasible set of task offers $\mathcal{{O}}^{\mathrm{feas},i}_t$ based on acceptance counter $\gamma^{\mathrm{a},i}_{k,z}$ and its perception $\theta^{-i}_i$ about other \glspl{mcsp}.
% Part 1
Considering $\gamma^{\mathrm{win},i}_{k,z}$, an acceptance ratio of every $(k,z)$ is evaluated.
From this, the \gls{mcsp} $i$ evaluates expected utility given the acceptance ratio as $\Hat{U}^{\mathrm{win},i}_{k,n,t} = \Hat{U}^{\mathrm{MCSP},i}_{k,n,t} \frac{\gamma^{\mathrm{win},i}_{k,z}}{(\gamma^{\mathrm{win},i}_{k,z}+\gamma^{\mathrm{lost},i}_{k,z})}$.
The \gls{mcsp} $i$ prunes different $\langle k, z, p \rangle$ combinations for which the expected utility $\Hat{U}^{\mathrm{win},i}_{k,n,t} \leq 0$.
% Part 2
Additionally, the \gls{mcsp} $i$ prunes the payment indices based on $\theta^{-i}_i$.
To do so, it estimates the chances of winning, i.e, attract the \gls{mu} $k$ to perform the task $O^i_{n,t}$, based on the payment index $p$ and own perceptions $\theta^{-i}_i$ as $\text{Pr}^{\mathrm{win},i}(p)= \sum_{p'=0}^p\pi^{-i}_{k,z}(p')$ where $\pi^{-i}_{k,z}$ is a probability mass function of other \gls{mcsp} for the given \gls{mu} $k$ and task type $z$.
All payment levels with a probability of winning below a threshold are removed.
Then, for every \gls{mu}, task type, and the payment index in the feasible task offer set $\mathcal{{O}}^{\mathrm{feas},i}$, we compute the UCB value given by
\begin{equation}
    \mathrm{UCB}^i_{k,z,p} = \Bar{U}^{\mathrm{MCSP},i}_{k,n,t} + \mathrm{UCB}_c\sqrt{\frac{\log(\mathrm{UCB}_t)}{L_{k,z,p}}},
\end{equation}
where $\mathrm{UCB}_c$ and $\mathrm{UCB}_t$ denote the UCB exploration constant and UCB time step index, respectively.
$L_{k,z,p}$ monitors how often the contract $x^i_{k,z,p}$ has been selected (Line 7-15).
Initially, the algorithm explores different contracts to gather more information about them and as the time progresses, the algorithm exploits the gathered information.
The feasible task offer set ${\mathcal{O}^{\mathrm{feas},i}}$ is then sorted according to the decreasing order of UCB values (Line 16).
Out of this set, task offers are chosen respecting task quotas and ensuring that each \gls{mu} receives only one offer per \gls{mcsp}.
Therefore, if \gls{mu} $k$ is already assigned, i.e. $k \in \zeta^{\mathrm{a}}_t$, then we skip all contracts involving \gls{mu} $k$ (Line 18-19).
Similarly, if the task quotas are already exhausted, i.e., $\rho^{\mathrm{a},i}_z = \rho^i_z$, then we skip all the contracts which involve the task type $z$ (Line 20-21).
When a task offer $O^i_{k,n,t}$ is created from the feasible set, the values of $\zeta^\mathrm{a}_t$ and $\rho^{\mathrm{a},i}_z$ are updated (Line 24-26).
The \gls{mcsp} $i$ then sends all the task offers to the respective \glspl{mu} and waits for their responses (Line 28).
For every task offer, the \gls{mu} responds with either an accept or a reject decision.
If the task offer is accepted, i.e. $y^i_{k,n,t}=1$, the acceptance counter $\gamma^{\mathrm{win},i}_{k,z}$ is increased by one (Line 31).
\gls{mu} $k$ performs task $a_{n,t}^i$ and sends the task result back to \gls{mcsp} $i$.
From the result, \gls{mcsp} $i$ evaluates the true quality $q^i_{k,n,t}$ and observes $U^{\mathrm{MCSP},i}_{k,n,t}$ (Line 32).
From the observed utility $U^{\mathrm{MCSP},i}_{k,n,t}$, the estimated expected utility is updated (Line 33) as
\begin{equation}
\label{eq:pacmab_MCSP_utility_update}
    \Hat{U}^{\mathrm{MCSP},i}_{k,n,t} = \Hat{U}^{\mathrm{MCSP},i}_{k,n,t-1} + \frac{U^{\mathrm{MCSP},i}_{k,n,t} - \Hat{U}^{\mathrm{MCSP},i}_{k,n,t}}{L_{k,z,p}}.
\end{equation}
If the offer is rejected, i.e. $y^i_{k,n,t}=0$, then the rejection counter $\gamma^{\mathrm{lost},i}_{k,z}$ is increased by one and the expected utility from the previous time step is used again.

\begin{algorithm}[t]
    \footnotesize
    \caption{MU $k$'s learning-based task offer acceptance algorithm}\label{alg:MUs_acceptance_algorithm}
    \begin{algorithmic}[1]
        \STATE \textbf{Initialize:} $\hat{U}^\mathrm{MU}_{k,z}$ $\forall z \in \mathcal{Z}$.
        \FOR{t=1,\ldots, T}
            \STATE Draw $\epsilon^\mathrm{MU}_k$ from $\mathcal{U}[0,1]$.
            \IF{$\epsilon^\mathrm{MU}_k < \epsilon^a$}
                \STATE Select random task offer $\Hat{O}^i_{k,n,t}$ from the received offers $\forall i \in \mathcal{I}$.
                \STATE Convey the acceptance decision $y^i_{k,n,t}=1$ to the respective \gls{mcsp} $i$.
            \ELSE
                \STATE Select a task offer $\Hat{O}^i_{k,n,t}$ which maximizes the $\hat{U}^\mathrm{MU}_{k,n,t}$
                \STATE Convey the acceptance decision $y^i_{k,n,t}=1$ to the respective \gls{mcsp} $i$.
            \ENDIF
            \STATE Convey the rejection decision $y^{-i}_{k,n,t}=0$ to the other \glspl{mcsp} along with rejection reason.
            \STATE Perform the task $O^i_{n,t}$ and transmit the result $r_{k,n,t}$ to the \gls{mcsp} $i$.
            \STATE Receive the payment $P^i_{k,n,t}$ and observe the $U^\mathrm{MU}_{k,n,t}$, $\totalTime$, and $\totalEnergy$.
            \STATE Update the estimates $\hat{U}^\mathrm{MU}_{k,z}$. \hfill
            \algorithmiccomment{Eq.~(\ref{eq:pacmab_MU_utility_update})}
        \ENDFOR
    \end{algorithmic}
\end{algorithm}
The \glspl{mu} do not know the efforts required to perform different tasks in advance, and have to learn them over time.
To learn these task efforts, every \gls{mu} $k$ runs a multi-armed bandit in which it learns about the task efforts for each task type $z$.
\gls{mu} $k$'s learning-aided task acceptance algorithm is given in Algorithm~\ref{alg:MUs_acceptance_algorithm}.
Every \gls{mu} $k$ initializes its expected utility $\Hat{U}^\mathrm{MU}_{k,z}$ for every task type $z$ with zeros (Line 1).
In every time step $t$, every \gls{mu} $k$ draws a random variable $\epsilon^\mathrm{MU}_k$ between $[0,1]$.
If $\epsilon^\mathrm{MU}_k < \epsilon^a$, the \gls{mu} explores, else exploits (Line 3).
In the exploration phase, the \gls{mu} randomly selects a task offer from one  of the offering \glspl{mcsp} (Line 4).
In the exploitation phase, the \gls{mu} $k$ selects the task offer $\Hat{O}^i_{k,n,t}$  that maximizes its expected utility $\Hat{U}^\mathrm{MU}_{k,n,t}$ (Line 7).
After a decision has been made, \gls{mu} $k$ informs the respective \gls{mcsp} $i$ with an acceptance decision, i.e. $y^i_{k,n,t}=1$ (Line 8).
\gls{mu} $k$ conveys its rejection decision along with the reason to all the other \glspl{mcsp} (Line 9).
If the task is accepted, \gls{mu} $k$ performs task $O^i_{n,t}$ and transmits the result $r_{k,n,t}$ back to \gls{mcsp} $i$.
Afterwards, it receives payment $P^i_{k,n,t}$ and observes  $U^\mathrm{MU}_{k,n,t}$.
The expected utility estimate is updated as
\begin{equation}
\label{eq:pacmab_MU_utility_update}
    \Hat{U}^\mathrm{MU}_{k,z} = \Hat{U}^\mathrm{MU}_{k,z} + \frac{(U^\mathrm{MU}_{k,n,t} -\Hat{U}^\mathrm{MU}_{k,z})}{M_z},
\end{equation}
where $M_z$ represents the number of times \gls{mu} $k$ has performed tasks of type $z$.

\vspace{-4mm}
\subsection{Stability and convergence of PACMAB}
In this section, we discuss the stability and convergence properties of the proposed PACMAB algorithm.
As is common in multi-agent reinforcement learning (MARL) frameworks, deriving formal theoretical guarantees for stability and convergence is challenging.
This difficulty is compounded in the case of PACMAB, which operates within a MARL framework wherein heterogeneous learning agents (MABs) interact in a competitive setting \cite{liu2021bandit, zhang2021multiagentreinforcementlearningselective}.
In such environments, even simple example instances can give rise to considerable analytical complexity, rendering closed-form convergence proofs intractable \cite{zhang2021multiagentreinforcementlearningselective}.
Nevertheless, through extensive numerical evaluation, we empirically demonstrate that the PACMAB algorithm converges to the solution obtained by the PRISM algorithm, which has been formally shown to attain the stable hyper-Nash equilibrium (SHNE), i.e., stable assignments under \gls{mcsp}-specific perceptions.

\vspace{-5mm}
\subsection{Computational complexity analysis}
Since PACMAB is fully decentralized, we analyze the complexity from the perspective of the entity that runs the algorithm.
As \gls{mcsp} $i$ runs a combinatorial UCB algorithm (Algorithm~\ref{alg:proposed_algorithm}) in every time step $t$, we first take a look at the worst case complexity.
In lines 4-15 we can see that the algorithm computes a UCB value for each arm with complexity $O(1)$.
For $K$ \glspl{mu}, $Z$ task types, and $P$ payment levels, the worst case complexity of computing the UCB value is $O(KZP)$.
Afterwards, in line 16, these UCB values are sorted, which has complexity of $O(KZP\log(KZP))$ \cite{Zhang_2024_MAB_UCB_FL}.
The rest of the algorithm has complexity of $O(N)$ since it involves updating the UCB values of the $N$ selected arms.
Thus, the total computational complexity of the algorithms is given as $O(KZP)+O(KZP\log(KZP))+O(N) = O(KZP \log(KZP))$ because the dominant term is $O(KZP \log(KZP))$ as $ZP\gg\log N$ in our case.
Thus, over the time horizon $T$, the complexity becomes $O(TKZP)$.
This complexity is reasonable for the MCSP since it has linear dependence on the time horizon, number of participating \glspl{mu}, number of available task types, and the payment levels.
MCSPs are equipped with sufficient computational capacity to implement an algorithm such as PACMAB.

At the \gls{mu}-side, Algorithm~\ref{alg:MUs_acceptance_algorithm} is used to learn about $Z$ different task types by selecting one task at a time.
The algorithm computes the expected utility of each task type $z$ by performing a task and updating the estimate using the sampled efforts. Such operation has complexity of $O(1)$ (Line 4-13).
Since each \gls{mu} has to decide whether to accept the task or not, out of maximum $I$ offers, the algorithm in time step $t$ has computing complexity of $O(I)$ \cite{Bernd_WSaad_OSL_2024}.
Over the entire time horizon $T$, the resulting complexity
is $O(TI)$.
This complexity is reasonable since the MUs are typically simple devices with limited computational capacity.
Our proposed PACMAB algorithm respects this constraint and enables the MUs to make decisions which maximize their achieved utilities at a low computational cost.

Note that for both, the \glspl{mcsp} and the \glspl{mu}, the communication overhead required for matching is low.
The \gls{mcsp} sends task offers to each \gls{mu} which contains only the task type $z$ and the payment information.
The task acceptance as well as the task rejection with reason is a short message which the \glspl{mu} transmit back to the respective \glspl{mcsp}.
The \glspl{mu} then perform the accepted task and transmit the result back to the respective \glspl{mcsp}.
\vspace{-4mm}
\section{Simulation Results and Analysis}
\label{sec:numerical_evaluation}

\begin{figure*}[t]
\centering
\vspace{-15mm}
% ---------- Legend ----------
\includegraphics[width=0.9\linewidth]{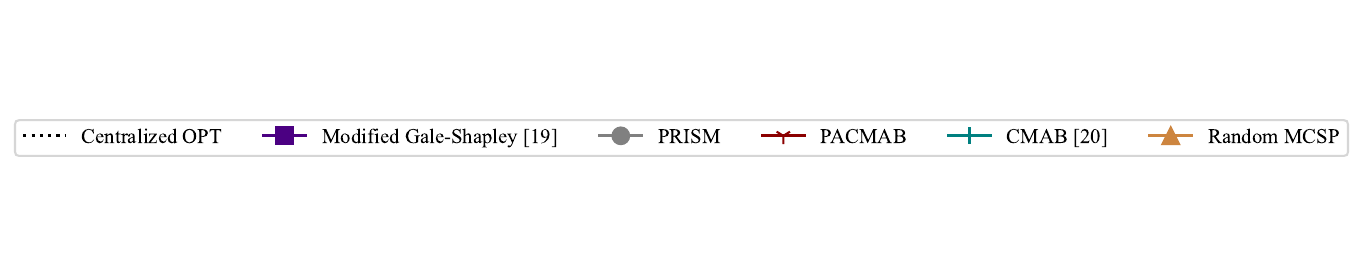}
\vspace{-16mm} % Next blank line is important

% ---------- Row 1: (a)(b)(c) ----------
\subfloat[\footnotesize Achieved average social welfare]{
  \includegraphics[width=0.32\linewidth, trim=0 8pt 0 0, clip]{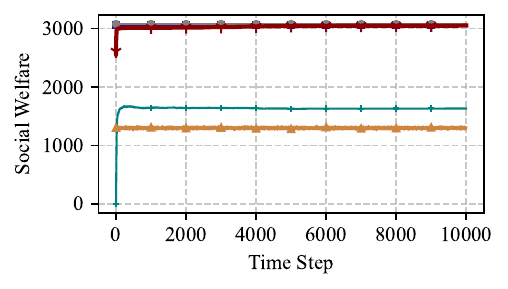}
  \label{fig:system_social_welfare}
}
\subfloat[\footnotesize Achieved average MCSP utility]{
  \includegraphics[width=0.32\linewidth, trim=0 8pt 0 0, clip]{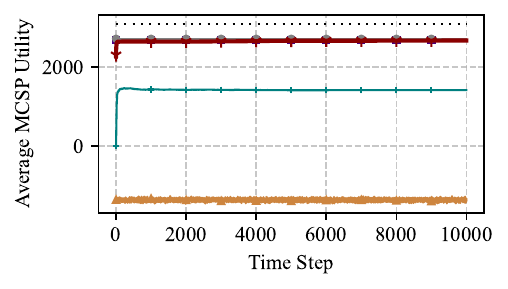}
  \label{fig:system_mcsp_utility}
}
\subfloat[\footnotesize Achieved average MU utility]{
  \includegraphics[width=0.32\linewidth, trim=0 8pt 0 0, clip]{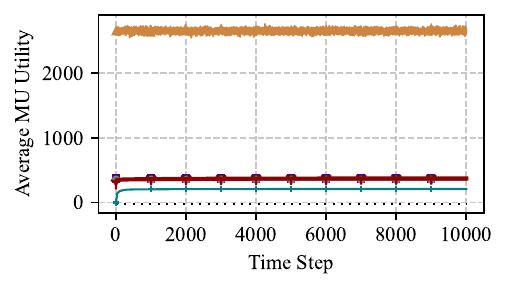}
  \label{fig:system_mu_utility}
}
\vspace{-5mm}
\par\medskip %
% ---------- Row 2: (d)(e)(f) ----------
\subfloat[\footnotesize Average completed tasks]{
  \includegraphics[width=0.32\linewidth, trim=0 8pt 0 0, clip]{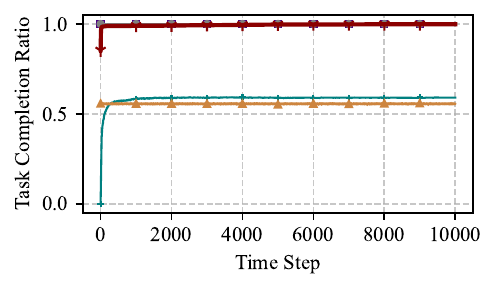}
  \label{fig:system_completed_tasks}
}
\subfloat[\footnotesize Average cumulative collisions]{
  \includegraphics[width=0.32\linewidth, trim=0 8pt 0 0, clip]{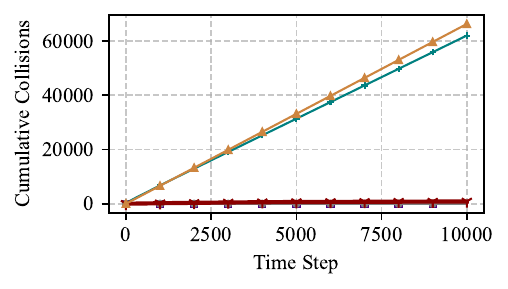}
  \label{fig:system_cumulative_collisions}
}
\subfloat[\footnotesize Average energy consumption]{
  \includegraphics[width=0.32\linewidth, trim=0 8pt 0 0, clip]{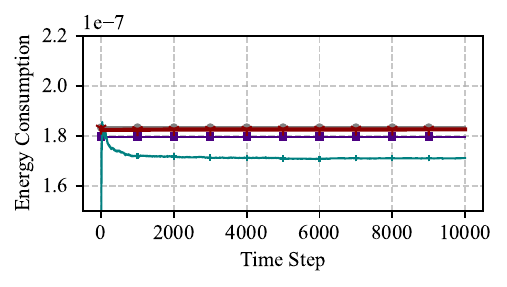}
  \label{fig:system_energy_consumption}
}
\caption{Performance comparison for $K=50, N=[10,50], Z=5$.}
\label{fig:baseline_performance}
\vspace{-5mm}
\end{figure*}

\subsection{Evaluation metrics}
\label{subsec:evaluation_metrics}
Since the \glspl{mcsp} and the \glspl{mu} have conflicting interests, we consider different metrics to evaluate the performance of our proposed PACMAB algorithm.
Specifically,  we consider metrics evaluating the complete \gls{mcs} system as well as metrics from the perspective of the \glspl{mcsp} and \glspl{mu}.

a) \textit{Social welfare}: Social welfare is often used to evaluate the collective performance of the \gls{mcs} system \cite{simon2023decentralized}.
Mathematically, the social welfare $U^{\mathrm{SW}}_t(\textbf{Y}_t)$ is given by
\begin{equation}
\label{eq:social_welfare}
    U^{\mathrm{SW}}_t(\textbf{Y}_t) = \sum_{i=1}^I \sum_{k=1}^K \sum_{n=1}^{N_i} y_{k,n,t}(U^{\mathrm{MCSP},i}_{k,n,t} + U^{\mathrm{MU}}_{k,n,t}).
\end{equation}

b) \textit{Achieved average \gls{mcsp} and \gls{mu} utility}: 
To study the effectiveness of the algorithms from the individual entity's perspective, we consider individual achieved utilities of \glspl{mcsp} and \glspl{mu}.
This is the net revenue earned by the \glspl{mcsp} and the \glspl{mu} individually.

c) \textit{Task completion ratio}: This is the ratio of the total tasks available to the total number of tasks completed in a time step.

d) \textit{Cumulative collisions}: A collision is the event when different \glspl{mcsp} send a task offer to the same \gls{mu}.
Since the \gls{mu} can only accept one offer, the rejected offers remain unfinished.
Thus, collisions degrade the performance of the \gls{mcs} system since the rejected offers affect the task completion of the respective \glspl{mcsp}.

\vspace{-3mm}
\subsection{Baseline algorithms}
\label{subsec:baseline_algorithms}
We use the following benchmark algorithms to compare the performance of our proposed algorithm.
\begin{itemize}
    \item \textit{Centralized OPT} (COPT): This is an offline optimization-based solution which requires complete information about the entire \gls{mcs} system, i.e., the qualities and efforts of the \glspl{mu} for every task type.
    The algorithm aims to find a task assignment that maximizes the social welfare given in (\ref{eq:social_welfare}).
    Note that the algorithm ignores the individual preferences of the \glspl{mcsp} and the \glspl{mu}.

    \item \textit{Perception-aware matching solution} (PRISM): This is the offline game-theory-based solution presented in Section \ref{sec:solution_algorithms}.
    The algorithm requires complete information about the individual preferences of each player and iteratively finds a solution to the dynamic hypergame.
    Using this algorithm, every player tries to maximize its own utility selfishly.

    \item \textit{Modified Gale-Shapley} (MGS) \cite{Bernd_WSaad_OSL_2024}: This is an offline game-theory-based solution which uses the well-known deferred acceptance (DA) algorithm to iteratively find stable task assignments.
    Due to the restriction in our \gls{mcs} system, one \gls{mcsp} cannot offer multiple tasks to the same \gls{mu}.
    Thus, we modify the implementation of DA in \cite{Bernd_WSaad_OSL_2024} to an \gls{mcsp} proposing scenario and enforce the one task offer per \gls{mu} constraint.
    The algorithm requires complete information about the preferences of both, the \glspl{mcsp} and the \glspl{mu}.

    \item \textit{CMAB} \cite{C1:own:Dongare2024b}: This two-sided learning approach uses a combinatorial upper confidence bound algorithm at each \gls{mcsp} without considering the perceptions of the other \glspl{mcsp}. At the \gls{mu}-side, a multi-armed bandit is implemented which learns the efforts of different task types.

    \item \textit{Random MCSP}: This is a benchmark algorithm which requires no information about the \gls{mcs} system. The \glspl{mcsp} randomly offer tasks to the \glspl{mu} with random payments.
    The \glspl{mu} strategically choose the better offers which maximize their own utility.
\end{itemize}

\subsection{Simulation setup}
\label{subsec:simulation_setup}
For the numerical evaluation, we consider $100$ independent Monte Carlo iterations.
Each iteration runs for $T=10000$ time steps.
The number of \glspl{mcsp} is set to $I=2$, i.e., \gls{mcsp} $0$ and \gls{mcsp} $1$.
The number of available \glspl{mu} is set to $K=50$, and the number of available tasks per \glspl{mcsp} varies between $10 \leq N_i \leq 50$ tasks in each time step for each \gls{mcsp}.
We consider $Z=10$ types of tasks \cite{Bernd_WSaad_OSL_2024}.
Rest of the simulation parameters are summarized in Table \ref{tab:evaluation_parameters}.
\begin{table}
  \caption{Simulation parameters}
  \label{tab:evaluation_parameters}
	 \centering{
  \scriptsize
	 \def\arraystretch{1.2}
	\begin{tabular}{|l|l|}
	\hline
	\textbf{Parameter} & \textbf{Value}  \\
	\hline
	Total number of time steps $T$ & $10000$ time steps \\
	Number of available MUs $K$ & $K=[50, 200]$\\
    Number of available task types \cite{Bernd_WSaad_OSL_2024} & $Z=[5, 25]$ tasks \\
    Number of payment levels & $|\mathcal{P}^i_z|= 20$ levels\\
    Number of tasks per type & $|\mathcal{O}^i_z|=[1,5]$ tasks \\
    Mean communication rate \cite{Bernd_WSaad_OSL_2024} & $\Bar{\tau}^\mathrm{comm}_{k,z}=[40, 80] \SI{}{\mega\bit/\second}$ \\
    CPU frequency \cite{Bernd_WSaad_OSL_2024} & $f^\mathrm{local}_k = [1, 2] \SI{}{\giga\hertz}$ \\
    Mean sensing time \cite{Mahn2021globalOrchestration} & $\Bar{\tau}^\mathrm{sense}_{k,z}=[60, 180] \SI{}{\milli\second}$ \\
    Transmission power \cite{Bernd_WSaad_OSL_2024} & $p^\mathrm{comm}_k=200 \SI{}{\milli\watt}$ \\
    Computing power \cite{Bernd_WSaad_OSL_2024} & $\SI{1}{\watt}$ \\
    Computational complexity \cite{OPAT_Huang_2022} & $c_z=[200,300]$ CPU cycles/bit \\
    Sensing data size \cite{Bernd_WSaad_OSL_2024} & $d_z = [50,100]\SI{}{\mega\bit}$ \\
    Sensing result size \cite{Bernd_WSaad_OSL_2024} & $s_z = [10, 20]\SI{}{\mega\bit}$ \\
    \hline
    \gls{mu}'s time cost parameter \cite{Bernd_WSaad_OSL_2024, C1:own:Dongare2024b} & $\alpha_k=0.01 \frac{\text{Monetary units}}{s}$ \\
    \gls{mu}'s energy cost parameter \cite{Bernd_WSaad_OSL_2024, C1:own:Dongare2024b} & $\beta_k=0.004 \frac{\text{Monetary units}}{J}$ \\
    \hline
    UCB exploration constant & $\text{UCB}_c=2$ \\
    \gls{mu} exploration constant & $\epsilon^a=1$ \\
    \gls{mu} exploration rate & $\epsilon_t = 0.999$ \\
    %Grid resolution & [200]\SI{}{\metre} \\
	\hline
	\end{tabular}}
\end{table}

\subsection{Results and discussion}
\label{subsec:results}

\begin{figure}[t]
   \centering
    \begin{minipage}[c]{0.48\linewidth}
        \includegraphics[width=\linewidth]{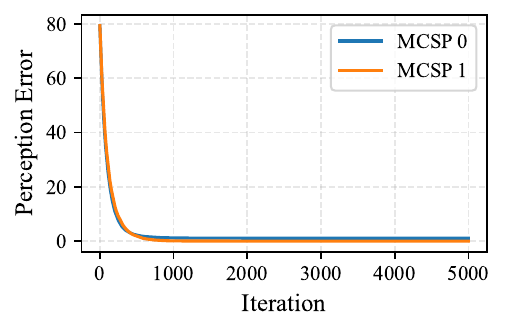}
        \vspace*{-7mm}
        \caption{PRISM: Perception error over iterations}
        \label{fig:hypergame_perception_error_evolution}
    \end{minipage}
    \hfill
    \begin{minipage}[c]{0.49\linewidth}
        \includegraphics[width=\linewidth]{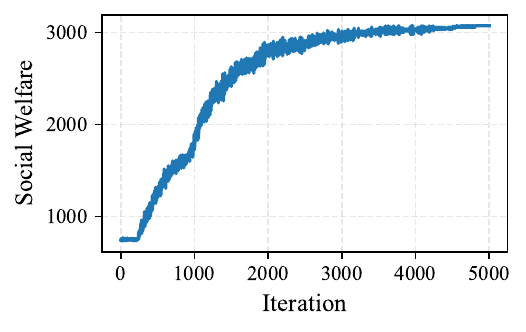}
        \vspace*{-7mm}
        \caption{Achieved social welfare of PRISM over iterations}
        \label{fig:hypergames_social_welfare_evolution}
    \end{minipage}
\end{figure}
In Fig.~\ref{fig:system_social_welfare}, we compare the social welfare achieved by different benchmarks over time.
The COPT algorithm achieves the maximum social welfare by exploiting complete system information.
PRISM converges to COPT as misperceptions diminish through repeated iterations.
Both, COPT and PRISM, find the optimal assignments, however, in COPT, the \glspl{mu} do not receive any payments, therefore, the achieved utilities of the \glspl{mcsp} and the \glspl{mu} are different, but the sum, i.e., the social welfare, converges.
The evolution of the perception error is illustrated in Fig.~\ref{fig:hypergame_perception_error_evolution}.
As the perceptions improve, PRISM is able to improve its social welfare as shown in Fig.~\ref{fig:hypergames_social_welfare_evolution}.
In Fig.\ref{fig:system_social_welfare}, our PACMAB algorithm attains $99.1\%$ of the COPT performance, demonstrating that perception-aware learning effectively maximizes social welfare.
MGS achieves about $99.4\%$ of the COPT social welfare. 
However, its strategy to completely outbid the other \gls{mcsp} leads to suboptimal but unchangeable assignments.
In contrast, CMAB achieves only $52.9\%$ due to slow learning in a large combinatorial action space and convergence to local optima.
Random MCSP performs worst, reaching $42.3\%$, as it ignores both MCSP preferences and competition.

Figure~\ref{fig:system_mcsp_utility} shows the achieved \gls{mcsp}-side utility.
COPT attains the highest \gls{mcsp} utility as it exploits the complete information without considering the individual preferences of the \glspl{mcsp} and the \glspl{mu}.
PRISM and MGS achieve $87.9\%$ and $86.8\%$ of the COPT utility, respectively.
Note that the PRISM and MGS both exploit the complete information about the MCS system and, they also consider the individual preferences of the \glspl{mcsp} and the \glspl{mu}.
MGS performs worse than PRISM because its contest-based mechanism overpays the \glspl{mu} such that they will accept the offer, which results in lower \gls{mcsp} utility.
PACMAB achieves $86.4\%$ of the COPT utility without requiring complete information by leveraging perceptions to prune contracts which lead to low achieved utility and focus on the ones which the \glspl{mcsp} as well as the \glspl{mu} prefer.
In contrast, CMAB reaches only $45.9\%$ due to the lack of perception-aware learning.
Random MCSP results in negative utility as it does not use any information.

In Fig.~\ref{fig:system_mu_utility}, we see that the achieved MU utility in case of the COPT is negative.
This is because COPT forces the MUs to perform the tasks without any payments.
Our proposed PACMAB algorithm achieves higher \gls{mu} utility and converges to PRISM and MGS solution such that the \glspl{mu} are also satisfied from the assignments.
On the contrary, the Random \gls{mcsp} overpays the \glspl{mu} and thus achieves high MU utility.
However this is unrealistic since the achieved \gls{mu} utility is at the cost of \gls{mcsp} utility.
The CMAB achieves \gls{mu} utility of only $55.91\%$ as compared to our proposed PACMAB because CMAB performs less tasks and therefore, in CMAB, the \glspl{mu} earn less on average.

In Fig.~\ref{fig:system_completed_tasks}, we analyze task completion ratios of the benchmark solutions in comparison with our proposed PACMAB.
COPT and MGS algorithms exploit complete information and are able to complete all the available tasks.
When the misperceptions vanish, the PRISM algorithm also achieves a task completion ratio of $1$.
Our proposed PACMAB achieves $99.8\%$ task completion without the requirement of complete information about the MCS system.
This means, the PACMAB algorithm not only prioritizes high welfare, but also aims to maximize the task completion.
CMAB completes only $59\%$ of tasks since it fails to learn about the dynamic competition between the \glspl{mcsp} which results in poor performance.
The Random MCSP algorithm performs the worst by completing only $55.5\%$ of tasks since it does not utilize any information about the \gls{mcs} system.
The collision ratio, i.e., the ratio of rejected task offers over total offered tasks, directly affects the task completion ratio.

In Fig.~\ref{fig:system_cumulative_collisions}, we compare the average cumulative collisions over time.
These are task offer rejections in the scenario which degrade the task completion performance and consequently the achieved utilities of the \glspl{mu} and the \glspl{mcsp} along with the achieved social welfare.
The COPT and the MGS do not have any collisions since they exploit the complete information about the scenario.
PRISM minimizes the collisions by offering better task offer proposals over time and converges to COPT and MGS.
PACMAB learns \gls{mcsp}'s own preferences as well as reduces collisions over time.
Therefore, it exhibits sublinear cumulative collisions.
For the CMAB and the Random MCSP, the perceptions about the preferences of other MCSPs are not considered which results in frequent rejections of task offers.

In Fig.~\ref{fig:system_energy_consumption}, we compare the energy consumption of the benchmark algorithms in comparison with the proposed PACMAB algorithm.
We exclude the Random MCSP algorithm from the comparison for the clarity of the presentation.
The Random MCSP consumes high energy with a larger variance which obscures the performances of other schemes.
The CMAB algorithm consumes lower energy however, the algorithm also performs less tasks on average.
This result illustrates that our PACMAB algorithm achieves energy consumption which is comparable to that of the COPT and the PRISM without the requirement of the complete information.
The result highlights that the PACMAB algorithm achieves a superior performance while being energy efficient.

\begin{figure}[t]
    \centering
    \begin{minipage}[c]{\linewidth}
    \includegraphics[width=\linewidth]{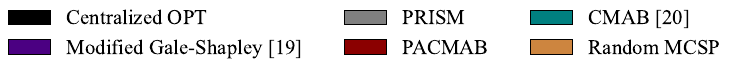}
    %\vspace{-13mm}
    \end{minipage}
    \hfill
    \begin{minipage}[c]{0.49\linewidth}
        \includegraphics[width=\linewidth]{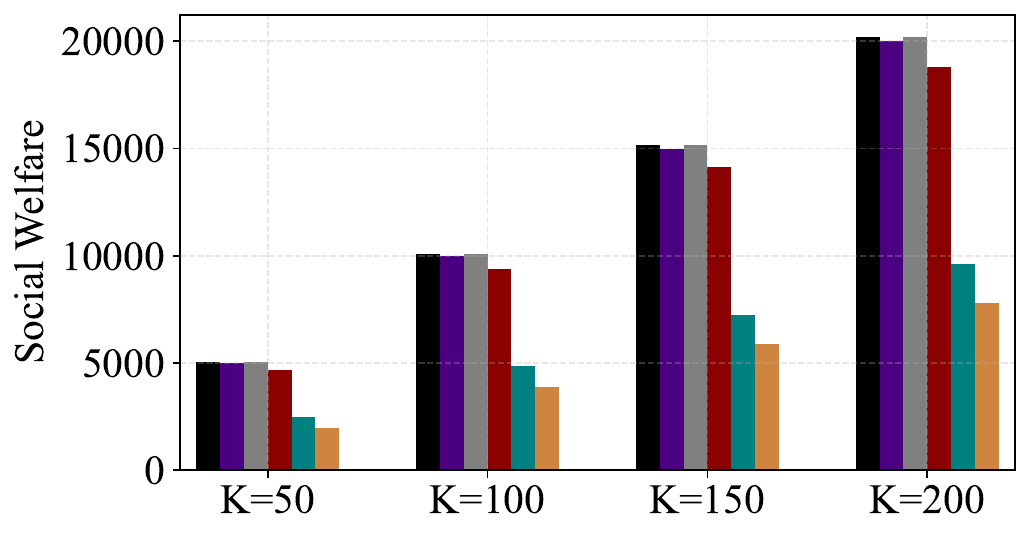}
        \caption{Achieved social welfare vs. number of MUs}
        \label{fig:mu_increasing_social_welfare}
    \end{minipage}
    \hfill
    \begin{minipage}[c]{0.49\linewidth}
        \includegraphics[width=\linewidth]{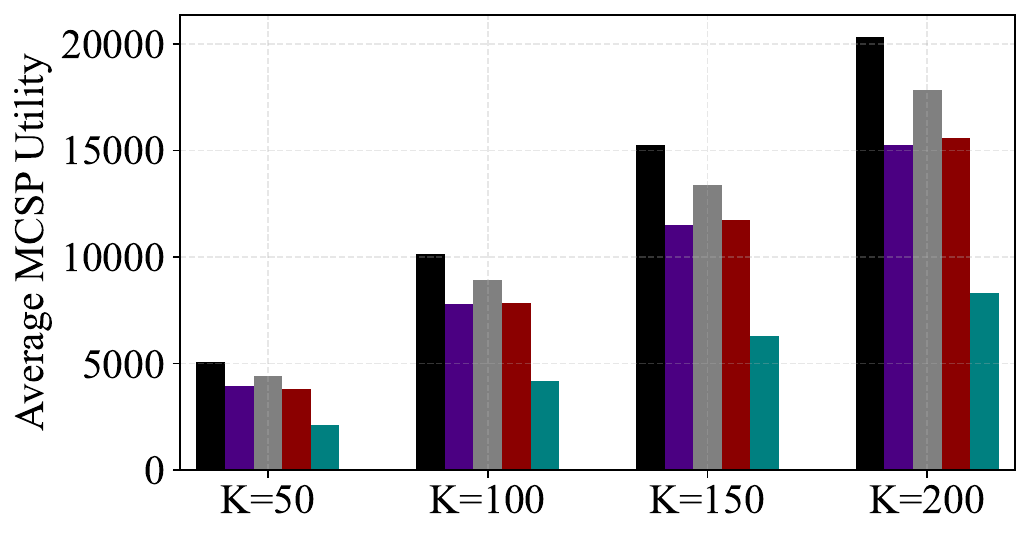}
        \caption{Achieved MCSP utility vs. number of MUs}
        \label{fig:mu_increasing_mcsp_utility}
    \end{minipage}
    \hfill
    \begin{minipage}[c]{0.49\linewidth}
        \includegraphics[width=\linewidth]{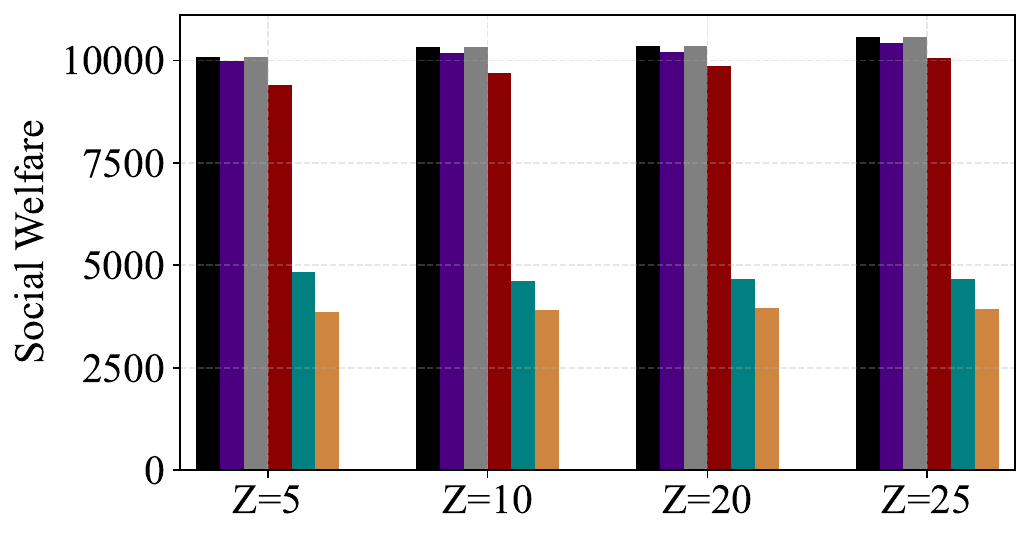}
        \caption{Achieved social welfare vs. number of task types}
        \label{fig:task_type_increasing_social_welfare}
    \end{minipage}
    \hfill
    \begin{minipage}{0.49\linewidth}
        \includegraphics[width=\linewidth]{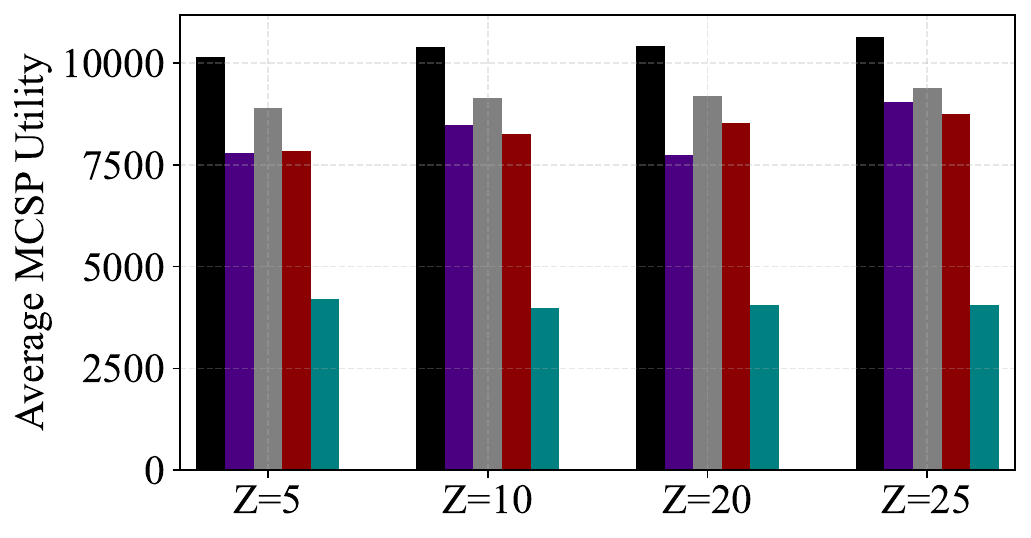}
        \caption{Achieved MCSP utility vs. number of task types}
        \label{fig:task_type_increasing_mcsp_utility}
    \end{minipage}
\end{figure}
To analyze the effect of increasing number of \glspl{mu} on the social welfare and the achieved MCSP utility, we consider $K=\{50,100,150,200\}$ \glspl{mu} and set $N = K$.
The results are illustrated in Fig.~\ref{fig:mu_increasing_social_welfare} and Fig.~\ref{fig:mu_increasing_mcsp_utility}, respectively.
As the number of \glspl{mu} increase, the achieved social welfare also increases.
For the case of $K=200$ \glspl{mu}, PRISM converges to the COPT algorithm while the MGS algorithm achieves $98.9\%$ social welfare.
In larger scenarios, the possible number of matching assignments grow exponentially with the number of MUs.
Consequently, it is difficult to learn which matching combinations are better suitable and which are not.
Still, PACMAB achieves at least $93.0\%$ social welfare as compared to the COPT, indicating that PACMAB is well-suited for larger networks with more \glspl{mu} and tasks.

To analyze the effect of heterogeneity of tasks on the social welfare and the achieved MCSP utility, we consider the following scenario.
The number of task types are varied between $Z=\{5, 10, 20, 25\}$.
We consider $K=N=100$ for this case.
The result of this analysis is shown in Fig.~\ref{fig:task_type_increasing_social_welfare} and in Fig.~\ref{fig:task_type_increasing_mcsp_utility}, respectively.
PACMAB achieves at least $95.0\%$ social welfare as compared to the COPT and PRISM algorithm.
As the number of tasks increase, the possible actions that each \gls{mcsp} can take also increase.
For example, for $K=100$ \glspl{mu}, $N=100$ tasks of $Z=25$ task types, and $20$ payment levels, the possible number of actions for each \gls{mcsp} are approximately $2^{297}$.
Even for such a high number, PACMAB learns efficient task proposal and task acceptance strategies.
In comparison, CMAB performs worse with at the most $46.3\%$ of social welfare and $46.8\%$ achieved \gls{mcsp} utility as compared to our PACMAB algorithm.
This demonstrates that the PACMAB algorithm is well-suited for MCS systems with high number of heterogeneous tasks, too.
\section{Conclusion}
\label{sec:conclusion}
In this paper, we have investigated competitive multi-platform mobile crowdsensing under incomplete information by modeling task offers and acceptances as a two-sided matching market with contracts.
To address uncertainty about competitors’ preferences, we have introduced a level-one dynamic hypergame formulation in which MCSPs update perceptions through repeated interactions and derived a perception-aware benchmark solution under partial-information assumptions.
To operate under fully unknown MU qualities and task execution efforts, we have proposed PACMAB, a fully decentralized perception-aware two-sided bandit-learning framework that learns task-proposal and task-acceptance strategies online.
PACMAB has linear computational complexity in terms of number of MUs, available tasks, and the discrete payment levels at the MCSP.
At the MU side, PACMAB exhibits linear complexity in terms of number of task offers received.
Simulation results demonstrate that PACMAB achieves at least $93\%$ of the optimal social welfare and over $99\%$ task completion, even as the number of MUs and task types scale significantly, and without assuming complete system information.
These findings confirm that perception-aware learning is a promising paradigm for decentralized MCS systems, effectively bridging the gap between fully informed centralized solutions and practical deployments under incomplete information.

\begin{spacing}{0.87}
\bibliographystyle{./bibliography/IEEEtran}
\bibliography{./bibliography/IEEEabrv,./bibliography/sample}
\end{spacing}
\appendices
\section{Proof of Theorem 1}
\label{proof_theorem1}
For property 1), consider the perception update rule from Algorithm 1 (Line 31) :
$\theta^i_{j,k,z}(t+1) = \max\{\theta^i_{j,k,z}(t), P^j_{k,n,t}\}$. We further define the perception gap for each $(j,k,z)$,
$$\delta^i_{j,k,z}(t) = \theta^j_{j,k,z} - \theta^i_{j,k,z}(t). \geq 0$$
We can make three key observations here. First, perceptions are monotonically non-decreasing,
$\theta^i_{j,k,z}(t+1) \geq \theta^i_{j,k,z}(t) \quad \forall t$. By Assumption~\ref{assump:rational}, MCSP $j$ will not pay more than its valuation (except during exploration). Therefore:
$P^j_{k,n,t} \leq \theta^j_{j,k,z} + \epsilon_t \cdot M$. When MCSP $i$ observes MU $k$ accepting an offer from MCSP $j$ at payment $P^j_{k,n,t}$, we have two cases. If $P^j_{k,n,t} > \theta^i_{j,k,z}(t)$ (a misprediction/surprise): \begin{align*} \delta^i_{j,k,z}(t+1) &= \theta^j_{j,k,z} - \theta^i_{j,k,z}(t+1)\\ &= \theta^j_{j,k,z} - \max\{\theta^i_{j,k,z}(t), P^j_{k,n,t}\}\\ &= \theta^j_{j,k,z} - P^j_{k,n,t}\\ &\leq \theta^j_{j,k,z} - \theta^i_{j,k,z}(t) = \delta^i_{j,k,z}(t).\end{align*}
The gap decreases and hence perception improves over time. Otherwise, if $P^j_{k,n,t} \leq \theta^i_{j,k,z}(t)$, i.e., no surprise, then
$$\theta^i_{j,k,z}(t+1) = \theta^i_{j,k,z}(t) \implies \delta^i_{j,k,z}(t+1) = \delta^i_{j,k,z}(t).$$
In both cases, $\delta^i_{j,k,z}(t+1) \leq \delta^i_{j,k,z}(t)$. Therefore:
$$\Delta\theta^i(t+1) = \sum_{j,k,z} \mathbb{E}[\delta^i_{j,k,z}(t+1)] \leq \sum_{j,k,z} \mathbb{E}[\delta^i_{j,k,z}(t)] = \Delta\theta^i(t).$$

For the property 2), we apply stochastic approximation theory by defining the Lyapunov function:
$$V^i(t) = \sum_{j \neq i} \sum_{k,z} (\theta^j_{j,k,z} - \theta^i_{j,k,z}(t))^2.$$
When an update occurs at time $t$ for triple $(j,k,z)$ (i.e., when MCSP $j$ competes for MU $k$ on task type $z$), we have:
\begin{align*}
\mathbb{E}[V^i(t+1) | \mathcal{F}_t] &= V^i(t) - (\theta^j_{j,k,z} - \theta^i_{j,k,z}(t))^2\\
&\quad + \mathbb{E}[(\theta^j_{j,k,z} - \max\{\theta^i_{j,k,z}(t), {P}^j_{k,z}(t)\})^2 | \mathcal{F}_t],
\end{align*}
where the expectation is taken over the random payment $P^j_{k,n,t}$ that MCSP $j$ will offer at time $t$ (not yet observed) and $\mathcal{F}_t$ denotes the filtration (history/information set) containing all observations available to MCSP $i$ up to and including time $t$, including past perceptions $\{\theta^i_{j,k,z}(s)\}_{s \leq t}$, actions, and feedback from MUs. Under rational play (Assumption~\ref{assump:rational}) with exploration,
$$\mathbb{E}[{P}^j_{k,z}(t)] = (1-\epsilon_t) \cdot P^{j,*}_{k,z} + \epsilon_t \cdot \bar{P}^j_{k,z},$$
where $P^{j,*}_{k,z}$ is the optimal payment and $\bar{P}^j_{k,z}$ is the exploration distribution mean.
The optimal payments satisfy $P^{j,*}_{k,z} \leq \theta^j_{j,k,z}$ by rationality, as $\epsilon_t \to 0$,
$\mathbb{E}[{P}^j_{k,z}(t)] \to \theta^j_{j,k,z}$.
This shows that the updates are in the direction of the true values. By the Robbins-Monro theorem \cite{robbins1951stochastic}, since Assumption~\ref{assump:exploration} ensures $\sum_t \epsilon_t = \infty$ and $\sum_t \epsilon_t^2 < \infty$, we have,
$$\theta^i_{j,k,z}(t) \to \theta^j_{j,k,z} - \delta^i_{j,k,z,\infty},
$$
where $\delta^i_{j,k,z,\infty} \geq 0$ is the residual error which is potentially zero. Therefore,
$$\Delta\theta^i(t) \to \Delta\theta^i_\infty = \sum_{j,k,z} \delta^i_{j,k,z,\infty}.
$$
%\vspace{-2mm}
For the property 3), we consider the exploitation phase when $\epsilon_t \approx 0$. The perception update can be approximated as a linear dynamical system. For each triple $(j,k,z)$, let $\gamma_{j,k,z}$ denote the probability that MCSP $j$ competes for $(k,z)$ and MCSP $i$ observes this event in a given time step. Then:
$$\mathbb{E}[\theta^i_{j,k,z}(t+1) - \theta^j_{j,k,z}] = (1 - \gamma_{j,k,z}) \cdot (\theta^i_{j,k,z}(t) - \theta^j_{j,k,z}) + \nu_{j,k,z}(t)$$
where $\nu_{j,k,z}(t)$ is exploration noise. Vectorizing, let $\vec{\delta}^i(t) \in \mathbb{R}^{KZ}$ contain all perception errors:
$$\vec{\delta}^i(t+1) = (I - \Gamma) \vec{\delta}^i(t) + \vec{\nu}(t)$$
where $\Gamma = \text{diag}(\gamma_{j,k,z})$ is a diagonal matrix with entries in $(0,1)$. The eigenvalues of $(I - \Gamma)$ lie in $[0, 1)$. Let $\lambda_{\max} = \max\{\text{eigenvalues of } (I-\Gamma)\} < 1$. Then:
\begin{equation}\mathbb{E}[\|\vec{\delta}^i(t)\|] \leq \lambda_{\max}^t \|\vec{\delta}^i(0)\| + \sum_{s=0}^{t-1} \lambda_{\max}^{t-s} \mathbb{E}[\|\vec{\nu}(s)\|]. \label{eq_dynamic_system}
\end{equation}
\eqref{eq_dynamic_system} follows from \cite{horn2012matrix} by solving the linear difference equation $\vec{\delta}^{\,i}(t+1) = (I-\Gamma)\vec{\delta}^{\,i}(t) + \vec{\nu}(t)$ by recursive substitution, which gives $\vec{\delta}^{\,i}(t) = (I-\Gamma)^t \vec{\delta}^{\,i}(0) + \sum_{s=0}^{t-1} (I-\Gamma)^{t-s-1} \vec{\nu}(s)$. Taking norms and using the bound $\|(I-\Gamma)^t\| \le \lambda_{\max}^t$, where $\lambda_{\max}$ is the spectral radius of $I-\Gamma$, yields the desired inequality.
Under Assumption~\ref{assump:exploration}, $\mathbb{E}[\|\vec{\nu}(s)\|] \leq C \epsilon_s$ for some constant $C$. Since $\sum_{s=0}^{\infty} \epsilon_s^2 < \infty$:
$$\sum_{s=0}^{\infty} \lambda_{\max}^{-s} \mathbb{E}[\|\vec{\nu}(s)\|] \leq C \sum_{s=0}^{\infty} \epsilon_s \lambda_{\max}^{-s} < \infty$$
Define $\lambda = -\log(\lambda_{\max}) > 0$. Then $\lambda_{\max} = e^{-\lambda}$, and:
$$\mathbb{E}[\Delta\theta^i(t)] \leq C_1 \|\vec{\delta}^i(0)\| \cdot e^{-\lambda t} + \frac{\epsilon_{\text{noise}}}{\lambda}$$
where $\epsilon_{\text{noise}} = C_2 \sup_t \epsilon_t$ for appropriate constants $C_1, C_2$.\\
For property 4), from the $\epsilon$-robust HNE analysis, the utility function satisfies a Lipschitz condition with respect to perception errors. Specifically, for any assignment $Y^i$ based on perception $\theta^i_{-i}$ and the optimal assignment $Y^{i*}$ based on true values $\theta^{-i}_{-i}$,
$$|U^{\text{MCSP},i}(Y^i; \theta^{-i}_{-i}) - U^{\text{MCSP},i}(Y^{i*}; \theta^{-i}_{-i})| \leq L \cdot \|\theta^i_{-i} - \theta^{-i}_{-i}\|_\infty,$$ where $L = K \cdot \max_z \max_{P \in \mathcal{P}^i_z} P$ is the maximum total payment difference across all MUs.
As $t \to \infty$,
$\|\theta^i_{-i}(t) - \theta^{-i}_{-i}\|_\infty \to \Delta\theta^i_\infty$. At SHNE, each MCSP plays optimally with respect to its (converged) perceptions. The utility achieved is:
$U^{\text{MCSP},i}_{\text{SHNE}} = U^{\text{MCSP},i}_{\text{opt}} - L \cdot \Delta\theta^i_\infty,$ where $U^{\text{MCSP},i}_{\text{opt}}$ is the utility at Nash equilibrium with perfect information. 
Therefore:
$$\liminf_{t \to \infty} \mathbb{E}[U^{\text{MCSP},i}_t] \geq U^{\text{MCSP},i}_{\text{SHNE}},$$
which can also be written as:
$$\liminf_{t \to \infty} \mathbb{E}[U^{\text{MCSP},i}_t] \geq U^{\text{MCSP},i}_{\text{opt}} - L \cdot \Delta\theta^i_\infty = U^{\text{MCSP},i}_{\text{SHNE}}$$

\end{document}